\documentclass{article}

\usepackage{arxiv}

\usepackage[utf8]{inputenc} 
\usepackage[T1]{fontenc}    
\usepackage{hyperref}       
\usepackage{url}            
\usepackage{booktabs}       
\usepackage{amsfonts}       
\usepackage{nicefrac}       
\usepackage{microtype}      
\usepackage{amsthm}
\usepackage{amssymb}

\theoremstyle{plain}
\newtheorem{proposition}{Proposition}
\newtheorem{corollary}{Corollary}

\usepackage{graphicx}
\usepackage{natbib}
\usepackage{doi}
\usepackage{amsmath}
\usepackage{algorithm}
\usepackage{algpseudocode}
\usepackage{caption}
\usepackage{subcaption}

\newcommand{\Var}{\mathrm{Var}}
\newcommand{\Cov}{\mathrm{Cov}}
\newcommand{\tr}{\mathrm{Tr}}
\newcommand{\rank}{\mathrm{Rank}}

\title{Bias Correction for Relative Importance Measures via Doubly Stochastic Reallocation}

\author{ \href{https://orcid.org/0009-0002-4187-513X}{\includegraphics[scale=0.06]{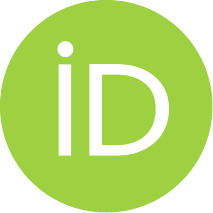}\hspace{1mm}Tien-En Chang} \\
	Institute of Industrial Engineering\\
	National Taiwan University\\
	Taipei, Taiwan \\
	\texttt{f09622009@ntu.edu.tw} \\
	\And
	\href{https://orcid.org/0000-0002-7951-9950}{\includegraphics[scale=0.06]{orcid.pdf}\hspace{1mm}Argon Chen}\thanks{Corresponding author.} \\
	Institute of Industrial Engineering\\
	National Taiwan University\\
	Taipei, Taiwan \\
	\texttt{achen@ntu.edu.tw} \\
}

\date{}

\renewcommand{\headeright}{}
\renewcommand{\shorttitle}{}
\newcommand{\ind}{\mathrel{\perp\!\!\!\perp}}

\begin{document}
\maketitle

\begin{abstract}

Relative importance (RI) analysis quantifies each predictor's contribution to the explained variance of a linear model. General Dominance (GD), a widely used benchmark, requires evaluating $2^p-1$ sub-models and becomes computationally intensive as the number of predictors $p$ grows. \textit{Orthogonalization-Reallocation Measures} (ORMs), including Relative Weights (RW) and the Green--Carroll--DeSarbo measure (GCD), provide efficient alternatives by assigning importance to orthogonalized predictors and reallocating it to the original predictors. Each, however, has a structural limitation: RW exhibits a \textit{leveling problem} that compresses differences among predictor importance values, whereas GCD exhibits an \textit{a priori bias} that systematically favors certain predictors before a response is observed. We show that this bias is governed by the row-sums of the reallocation matrix. A closed-form analysis under compound symmetry relates the reallocations underlying GCD and RW to a GD-based benchmark, showing that homogeneous multicollinearity alone does not induce an a priori bias and formalizing RW's leveling problem as excess shrinkage relative to the benchmark. We correct GCD's bias by mapping its reallocation matrix to a doubly stochastic matrix using the Method of Alternating Projections (MAP) and the Sinkhorn--Knopp (SK) algorithm, yielding GCD-MAP and GCD-SK. Comprehensive simulations show that GCD-SK removes the structural row-sum bias, substantially improves upon GCD, and often outperforms RW when the first principal component is dominant. We conclude with empirical guidelines for selecting among the measures.

\end{abstract}


\section{Introduction}
Relative importance (RI) analysis is critical for model interpretation in statistical modeling. Originating in quantitative behavioral and psychological research, RI defines a predictor's importance as its proportionate contribution to the explained variance of the response~\citep{johnson2004history,gromping2015variable}. Beyond mathematical attribution, RI allows decision makers to prioritize scarce organizational resources~\citep{johnson2004history}. For instance, in human resources, job supervisors utilize RI to determine which task or contextual performance dimensions are the primary drivers of overall job performance~\citep{johnson2001relative}. Similarly, in marketing analytics, practitioners rely on RI to infer which customer satisfaction attributes most significantly influence the overall customer experience~\citep{garver2020utilizing}.

Assessing importance is straightforward when predictors are uncorrelated, as simple metrics like squared marginal correlations or squared standardized regression coefficients suffice. However, the analysis becomes challenging in the presence of multicollinearity. When predictors are correlated, their contributions to the response partially overlap, rendering simple metrics ambiguous. A standard method to address this challenge is General Dominance (GD)~\citep{budescu1993dominance,azen2003dominance}, which provides a definitive solution by averaging a predictor's incremental contribution across all possible combinatorial sub-models. Due to its comprehensiveness and theoretical
rigor, GD is generally regarded as the most plausible measure of relative importance. However, its exponentially increasing complexity, requiring the evaluation of $2^p-1$ models for $p$ predictors, makes it computationally expensive when $p$ grows.

This computational barrier of GD has led to the development of \textit{Orthogonalization-Reallocation Measures} (ORMs), a class of efficient alternatives. These methods first transform correlated predictors into orthogonal predictors whose importance is unambiguous due to their orthogonality, and then reallocate the importance of orthogonal predictors back to the original predictors. Examples include importance weights ($w^2$)~\citep{johnson1966minimal}, the Green--Carroll--DeSarbo measure (GCD; $\delta^2$)~\citep{green1978new}, and Relative Weights (RW; $\varepsilon$)~\citep{johnson2000heuristic}. Among these, RW has gained popularity as a reliable and computationally efficient alternative to GD~\citep{leBreton2004monte,chao2008quantifying}. Consequently, RW has been applied across various contexts, including logistic regression~\citep{tonidandel2010determining} and multivariate regression~\citep{lebreton2008multivariate,hong2012dominance}. The statistical significance of RW has been examined in depth by~\cite{tonidandel2009determining}, while its extension to higher-order regression models is explored by~\cite{lebreton2013residualized}. More recently, RW has been adapted to high-dimensional settings, where the number of predictors exceeds the sample size, showing competitive performance compared to methods such as the lasso~\citep{tibshirani1996regression} in variable selection tasks~\citep{shen2020comprehensive,chang2026variable}. 

Despite this widespread adoption, the theoretical foundations of ORMs remain contested. Notably,~\cite{thomas2014johnson} argued that the original derivations of both RW and GCD are conceptually flawed and cautioned against their use, notwithstanding their close empirical agreement with GD. The ORM framework of~\cite{chang2025understanding} takes a different perspective: it sets aside the original rationales and views these methods simply as a common orthogonalization (Johnson's minimal transformation) paired with different reallocation rules---correlation-based for RW and regression-based for GCD---asking instead when and why such combinations approximate GD, a question that, beyond simulation evidence, the RI literature has largely left unexamined. From this viewpoint, the empirical success of ORMs is something to be explained, and their failures something to be located in specific predictor correlation structures, rather than grounds for dismissal.

Within this framework, two structural limitations have been identified~\citep{chang2025understanding}. RW suffers from the \textit{leveling problem}, which compresses differences among predictor importance values when the first principal component is dominant. GCD exhibits the \textit{a priori bias}, systematically favoring particular predictors before the response is observed---driven by row-sum imbalance in its reallocation matrix. The latter also carries a practical cost: ever since GCD was found to underperform~\citep{johnson2000heuristic}, applied research has relied almost exclusively on RW, leaving regression-based reallocation effectively unused even where leveling undermines RW. Correcting the a priori bias thus restores a competitive regression-based alternative precisely where RW is weakest. This paper makes three contributions:

\begin{itemize}
    \item We provide a self-contained formal characterization of the a priori bias. We show that the expected normalized importance assigned by an ORM is determined by the row-sums of its reallocation matrix, and relate GCD's row-sums to predictor Variance Inflation Factors (VIFs). In contrast, GD assigns every predictor an expected normalized importance of $1/p$. For GCD, we further show that these row-sums are weighted averages of predictor VIF ratios.
    \item We derive closed-form expressions for the GD-based (GDA), regression-based (RegPA, underlying GCD), and correlation-based (CorPA, underlying RW) reallocations under compound symmetry. This analytical benchmark shows that severe but homogeneous multicollinearity does not induce an a priori bias in GCD, while CorPA applies stronger shrinkage toward equal importance than GDA, providing an analytical illustration of RW's leveling problem.
    \item We correct the a priori bias by mapping the RegPA matrix to a nearby doubly stochastic matrix. We consider the Method of Alternating Projections (MAP)~\citep{von1949rings} and the Sinkhorn--Knopp (SK) algorithm~\citep{sinkhorn1967concerning}, yielding the corrected measures GCD-MAP and GCD-SK. Comprehensive simulations show that GCD-SK substantially improves upon uncorrected GCD and provides a robust alternative to RW across diverse predictor correlation structures.
\end{itemize}

The remainder of this paper is organized as follows. Section~\ref{sec: RI review} reviews GD and the ORM framework. Section~\ref{sec: a priori bias} characterizes the a priori bias, derives closed-form ORM expressions under compound symmetry, develops the proposed corrections, and provides a numerical illustration. Section~\ref{sec: simulations} evaluates the methods through comprehensive Monte Carlo simulations and presents practical guidelines. Finally, Section~\ref{sec: discussion and conclusion} summarizes the findings and discusses future research directions.
\section{Relative Importance Measures}\label{sec: RI review}
Let $X = [x_1, \ldots, x_p]$ be an $n \times p$ matrix of predictors and $y$ an $n \times 1$ response vector. We assume that $X$ has full column rank and that all predictors $x_i$ and the response $y$ are normalized to have zero mean and unit $\ell_2$ norm. 
\subsection{General Dominance}
General Dominance~\citep{budescu1993dominance,azen2003dominance} quantifies a predictor's importance by averaging its incremental contribution in predicting $y$ across all possible combinatorial sub-models, providing a comprehensive perspective. Specifically, the squared multiple correlation $R_{y\cdot X}^2$ is used to assess the contribution when all predictors $X$ are used in predicting $y$. Formally, GD of $x_i$ (with respect to $y$), denoted by $D_i$, can be defined as the weighted average of all possible increments:
\begin{equation}
    D_i=\frac{1}{p}\sum_{S\subseteq P\setminus\{i\}}\frac{1}{\binom{p-1}{|S|}}\left(R_{y\cdot X_{S\cup \{i\}}}^2-R_{y\cdot X_S}^2\right)
    \label{eq:GD definition}
\end{equation}
where $P=\{1,\ldots,p\}$ is a predictor index set, $S$ is a subset of predictor indices, $|\cdot|$ denotes the cardinality of a set, $S\subseteq P\setminus\{i\}$ denotes all possible subsets excluding the index of predictor $i$, and $X_S$ denotes a subset of predictors corresponding to indices $S$. It is worth noting that the concept of GD coincides with the Shapley value~\citep{shapley1953game}, which originates from game theory and has recently been used in the context of explainable machine learning. However, GD becomes computationally intensive as the number of predictors increases because it requires calculations for $2^p-1$ sub-models.

\subsection{Orthogonalization-Reallocation Measures}
\cite{chang2025understanding} introduced a unified framework for Orthogonalization-Reallocation Measures, decomposing them into two sequential steps: (1) Orthogonalization, and (2) Reallocation. In the first step, the original predictor matrix $X$ is transformed into an orthogonal predictor matrix $\tilde{Z}=[\tilde{z}_1,\ldots,\tilde{z}_p]$, satisfying $\tilde{Z}^\top\tilde{Z}=I_p$. Because the orthogonal predictors are uncorrelated, their individual contributions to the response $y$ are unambiguous and are quantified using their squared marginal correlations $\rho_{\tilde{z}_jy}^2$. These values define the importance of each orthogonal predictor.

In the second step, the importance of the orthogonal predictors is reallocated back to the original predictors. This process is governed by a reallocation matrix $A^{\tilde{Z}}\in\mathbb{R}^{p\times p}$, where each entry $a_{ij}^{\tilde{Z}}$ represents the proportion of importance from $\tilde{z}_j$ attributed to $x_i$. The resulting ORM importance estimate for each $x_i$ is given by:
\begin{equation}
    D_{A,i}^{\tilde{Z}}=\sum_{j=1}^p a_{ij}^{\tilde{Z}}\rho_{\tilde{z}_jy}^2\quad\text{and}\quad D_{A}^{\tilde{Z}}=[D_{A,1}^{\tilde{Z}},\ldots, D_{A,p}^{\tilde{Z}}]^\top.
    \label{eq:ORM}
\end{equation}

\paragraph{Orthogonalization Methods} Several orthogonalization methods are available, including principal component analysis (PCA), Gram--Schmidt, and the minimal transformation proposed by~\cite{johnson1966minimal}. The minimal transformation, in particular, minimizes the squared distance between the orthogonal predictors and the original predictors:
\begin{equation}
    O_J(\tilde{Z})=\tr\left((\tilde{Z}-X)^\top(\tilde{Z}-X)\right),
    \label{eq: objective function for the minimal transformation}
\end{equation}
where $\tr(\cdot)$ denotes the trace of a matrix. The minimizer of this optimization problem is 
\begin{equation}
    Z=X(X^\top X)^{-1/2}.
    \label{eq: the minimal transformation to orthonormality}
\end{equation}

\paragraph{Reallocation Methods} The reallocation step determines how the importance of orthogonal predictors is redistributed back to the original predictors. Based on the unified ORM formulation in Eq.~\eqref{eq:ORM},~\cite{chang2025understanding} summarized several reallocation methods, described below. These methods are defined for an arbitrary orthogonal predictor matrix $\tilde{Z}$:
\begin{itemize}
    \item Identity reallocation (IdA):
    \begin{align}
        a_{ij}^{\tilde{Z}}=\mathrm{IdA}_{ij}^{\tilde{Z}}=\begin{cases}
            1, &\text{if }i=j, \\
            0, &\text{otherwise.}
        \end{cases}
        \label{eq: IdA}
    \end{align}
    This method simply retains the orthogonal predictor's importance without redistribution. 

    \item Regression-based proportional reallocation (RegPA):
    \begin{align}
        a_{ij}^{\tilde{Z}}=\mathrm{RegPA}_{ij}^{\tilde{Z}}=\frac{\gamma^2_{\tilde{Z},ij}}{\sum_{k=1}^p \gamma^2_{\tilde{Z},kj}},
        \label{eq: RegPA}
    \end{align}
    where $\gamma_{\tilde{Z},ij}$ is the coefficient of $x_i$ in regressing $\tilde{z}_j$ on all $x_1,\ldots,x_p$. 

    \item Correlation-based proportional reallocation (CorPA):
    \begin{align}
        a_{ij}^{\tilde{Z}}=\mathrm{CorPA}_{ij}^{\tilde{Z}}=\frac{\ell^2_{\tilde{Z},ij}}{\sum_{k=1}^p \ell^2_{\tilde{Z},kj}},
        \label{eq: CorPA}
    \end{align}
    where $\ell_{\tilde{Z},ij}$ denotes the marginal (Pearson) correlation between $\tilde{z}_j$ and $x_i$. 

    \item GD-based reallocation (GDA):
    \begin{equation}
        a_{ij}^{\tilde{Z}}=\mathrm{GDA}_{ij}^{\tilde{Z}}=D_i(\tilde{z}_j, X),
        \label{eq: GDA}
    \end{equation}
    where $D_i(\tilde{z}_j, X)$ is the GD value of $x_i$ in regressing $\tilde{z}_j$ on all $x_1,\ldots,x_p$. This method was proposed by~\cite{chang2025understanding} as a theoretically ideal reallocation, as it best captures the contribution of $x_i$ to the variance explained in $\tilde{z}_j$. While GDA is computationally intractable for large $p$, it serves as a benchmark for evaluating the performance of practical ORM methods. 
\end{itemize}

\paragraph{Desirable Properties of Reallocation Matrices} 
Under the ORM framework, a reallocation matrix $A^{\tilde{Z}}$ should satisfy several key properties:
\begin{enumerate}
    \item Non-negativity: 
    \begin{equation}
        a_{ij}^{\tilde{Z}}\geq 0,\,\forall i,j=1,\ldots,p.
        \label{eq: ReA positive}
    \end{equation}
    This ensures interpretability, as negative allocations of importance are not meaningful.
    \item Column-sums equal to one: 
    \begin{equation}
        \sum_{i=1}^p a_{ij}^{\tilde{Z}}=1,\,\forall j=1,\ldots,p.    
        \label{eq: ReA col-sum 1}
    \end{equation}
    This ensures that the total importance of each orthogonal predictor is fully distributed across the original predictors, i.e., $\sum_{i=1}^{p}D_{A,i}^{\tilde{Z}}=\sum_{i=1}^p(\sum_{j=1}^p a_{ij}^{\tilde{Z}}\rho^2_{\tilde{z}_jy})=\sum_{j=1}^p\rho^2_{\tilde{z}_jy}=R^2_{y\cdot X}$.
    \item Row-sums equal to one: 
    \begin{equation}
        \sum_{j=1}^p a_{ij}^{\tilde{Z}}=1,\,\forall i=1,\ldots,p.
        \label{eq: ReA row-sum 1}
    \end{equation}  
    This condition will be shown in the next section to prevent an ORM from systematically favoring particular predictors on average across possible responses.
\end{enumerate}  
Non-negativity ensures that the reallocations remain interpretable, whereas unit column-sums preserve the total explained variance. As shown in the next section, unit row-sums characterize the absence of the a priori bias under the response-generating framework considered here. A nonnegative matrix with unit row- and column-sums is doubly stochastic.

Empirical studies by~\cite{chang2025understanding} show that the minimal transformation $Z$ yields the most robust results across a variety of settings. Therefore, this paper will adopt the minimal transformation $Z$ as the default orthogonalization method. For notational simplicity, we drop the superscript $\tilde{Z}$ when referring to the reallocation matrix and ORM. When paired with the minimal transformation, the most common ORMs adopted in the literature are summarized in Table~\ref{tab: existing ORMs}.
\begin{table}[ht]
    \centering
    \caption{Existing ORMs using the minimal transformation $Z$ as the orthogonalization method.}
    \label{tab: existing ORMs}
    \begin{tabular}{|c|c|c|}
    \hline
       Reallocation Method  & Corresponding ORM & Reference \\ \hline
       IdA  & $w^2$ &~\cite{johnson1966minimal,chang2026variable} \\ \hline
       RegPA & GCD; $\delta^2$ &~\cite{green1978new} \\ \hline
       CorPA & RW; $\varepsilon$ &~\cite{johnson2000heuristic,genizi1993decomposition,shen2020comprehensive} \\ \hline
    \end{tabular}
\end{table}
\section{The A Priori Bias and Its Correction}\label{sec: a priori bias}
This section first characterizes the a priori bias in terms of the expected importance assigned across possible responses and relates GCD's bias to VIF-induced row-sum imbalance in the RegPA matrix. We then use compound symmetry as an analytical benchmark to distinguish the a priori bias from multicollinearity itself and to compare the reallocation behavior of GDA, RegPA, and CorPA. Finally, we develop two corrections based on double stochasticity and illustrate their effects numerically.
\subsection{Characterization of the A Priori Bias}\label{subsec:a-priori-bias}

We characterize the a priori bias by considering expected importance across all possible systematic responses while holding the predictor correlation structure fixed. Throughout this section, a subscript $0$ denotes a random quantity. We first express the response in terms of an orthogonal basis:
\begin{equation}
    y_0=z_0^\top u_0+\epsilon_0,
    \label{eq:linear model (z)}
\end{equation}
where $z_0\sim N_p(0,I_p)$, $\epsilon_0\sim N(0,\sigma^2)$, and $z_{0i}\ind\epsilon_0$ for all $i$. The coefficient vector \[ u_0\sim\mathcal{U}(S^{p-1}), \qquad S^{p-1} = \left\{ u\in\mathbb{R}^p: \sum_{i=1}^p u_i^2=1 \right\}, \] is independent of $(z_0,\epsilon_0)$ and uniformly distributed on the unit sphere. Thus, the systematic component $\hat{y}_0=z_0^\top u_0$ ranges uniformly over all directions in the predictor space, while $\epsilon_0$ adds independent noise.

To obtain the corresponding model with correlated predictors, let $\Sigma\in\mathbb{R}^{p\times p}$ be a positive-definite predictor correlation matrix, and let $\Sigma^{1/2}$ denote its symmetric positive-definite square root. Define $x_0=\Sigma^{1/2}z_0$ and $\beta_0=\Sigma^{-1/2}u_0$ and the model becomes:
\begin{equation}
    y_0=x_0^\top\beta_0+\epsilon_0.
    \label{eq:linear model (x)}
\end{equation}
Here, $\beta_0$ is a derived random variable whose distribution is induced by the uniform distribution of $u_0\sim \mathcal{U}(S^{p-1})$. We formally write $\beta_0\sim F$, where $F$ denotes the distribution resulting from the transformation $\beta_0=\Sigma^{-1/2}u_0$. This setup of the relationship between $y_0$ and $x_0$ via $z_0$ is critical to ensure that the resulting systematic responses $\hat{y}_0=x_0^\top\beta_0=z^\top_0 u_0$ are uniformly distributed in the predictor space while having the correlation structure of $x_0$ specified by any $\Sigma$.

The RI measures discussed in the previous section are determined by the predictor correlation structure and the predictor--response correlations (equivalently, RI measures are functions of $\Sigma$ and $\beta_0$). To place all response realizations on a common scale, we normalize each RI vector by the full-model coefficient of determination $R^2_{y_0\cdot x_0}$. Accordingly, throughout this section, $D_i(\Sigma,\beta_0)$ and $D_{A,i}(\Sigma,\beta_0)$ denote the normalized GD and ORM importance values, respectively, each summing to one across predictors.

Considering all possible responses for a fixed set of predictors is particularly important in many practical contexts, such as behavioral science, where the same predictors are often used to analyze multiple responses. In this setting, an RI method should not systematically favor particular predictors before a response is observed. This principle motivates the concept of the a priori bias introduced by~\cite{chang2025understanding}.

\begin{proposition}[Expected normalized importance] \label{prop:expected-importance} 
Under the model in Eqs.~\eqref{eq:linear model (z)} and~\eqref{eq:linear model (x)}, the expected normalized GD importance satisfies 
\begin{equation} 
    \mathbb{E}_{\beta_0\sim F} \left[ D_i(\Sigma,\beta_0) \right] = \frac{1}{p}, \qquad i=1,\ldots,p. 
    \label{eq:expectation of GD} 
\end{equation} 
Moreover, let $A\in\mathbb{R}^{p\times p}$ be a fixed, response-independent reallocation matrix of an ORM based on a standardized orthogonal representation of the predictor space. Then 
\begin{equation} 
    \mathbb{E}_{\beta_0\sim F} \left[ D_{A,i}(\Sigma,\beta_0) \right] = \frac{1}{p}\sum_{j=1}^p a_{ij}, \qquad i=1,\ldots,p. \label{eq:expectation of ORMs} 
\end{equation} 
\end{proposition} 


The proof of Proposition~\ref{prop:expected-importance} is provided in Appendix~\ref{app:proof-expected-importance}. The proposition shows that the expected normalized importance assigned by an ORM is determined directly by the row-sums of its reallocation matrix. If the $i$-th row-sum is greater than one, the expected importance of $x_i$ exceeds the unbiased benchmark $1/p$; if it is less than one, the expected importance falls below $1/p$. Thus, the predictor correlation structure can predetermine which predictors are favored on average before any particular response is observed, giving rise to the a priori bias.

This result identifies unit row-sums as the condition for eliminating the a priori bias. We next examine why this condition can fail for GCD. Among the ORMs listed in Table~\ref{tab: existing ORMs}, GCD is the only method whose reallocation matrix does not generally have unit row-sums. The following proposition relates this failure to predictor VIFs under Johnson's minimal transformation, formalizing the mechanism identified by~\cite{chang2025understanding}.

\begin{proposition}[VIF-induced a priori bias of GCD]
\label{prop:gcd-vif-bias}
Let $z_0=\Sigma^{-1/2}x_0$ be Johnson's minimal transformation, and let
\[
    \Gamma_Z=(\gamma_{ij})_{i,j=1}^p=\Sigma^{-1/2},
    \qquad
    \mathrm{VIF}_j=(\Sigma^{-1})_{jj}.
\]
Then the RegPA weights used by GCD satisfy
\begin{equation}
    \mathrm{RegPA}_{ij}
    =
    \frac{\gamma_{ij}^2}{\mathrm{VIF}_j}.
    \label{eq:regpa-vif}
\end{equation}

Let $r_i=\sum_{j=1}^p\mathrm{RegPA}_{ij}$ denote the $i$-th row-sum and define
\[
    w_{ij}
    =
    \frac{\gamma_{ij}^2}{\mathrm{VIF}_i}.
\]
Then $w_{ij}\in[0,1]$, $\sum_{j=1}^p w_{ij}=1$, and
\begin{equation}
    r_i
    =
    \sum_{j=1}^p
    w_{ij}
    \frac{\mathrm{VIF}_i}{\mathrm{VIF}_j}.
    \label{eq:regpa-row-sum-vif}
\end{equation}

If $k\in\arg\max_{i=1,\ldots,p}\mathrm{VIF}_i$, then $r_k\geq1$, with strict inequality whenever there exists some $j$ such that $\gamma_{kj}\neq0$ and $\mathrm{VIF}_j<\mathrm{VIF}_k$. In that case,
\[
    \mathbb{E}_{\beta_0\sim F}
    \left[
        D_{\mathrm{RegPA},k}(\Sigma,\beta_0)
    \right]
    =
    \frac{r_k}{p}
    >
    \frac{1}{p}.
\]
Moreover, because RegPA is column-stochastic, $\sum_{i=1}^p r_i=p$. Hence, if $r_k>1$, then $r_\ell<1$ for at least one $\ell\neq k$, implying
\[
    \mathbb{E}_{\beta_0\sim F}
    \left[
        D_{\mathrm{RegPA},\ell}(\Sigma,\beta_0)
    \right]
    <
    \frac{1}{p}.
\]
\end{proposition}

The proof of Proposition~\ref{prop:gcd-vif-bias} is provided in Appendix~\ref{app:proof-gcd-vif-bias}. Equation~\eqref{eq:regpa-row-sum-vif} shows that each RegPA row-sum is a weighted average of the VIF ratios $\mathrm{VIF}_i/\mathrm{VIF}_j$. For a maximal-VIF predictor,
\begin{equation}
    r_k-1
    =
    \sum_{j=1}^p
    w_{kj}
    \left(
        \frac{\mathrm{VIF}_k}{\mathrm{VIF}_j}-1
    \right).
    \label{eq:max-vif-excess-row-sum}
\end{equation}
Thus, large VIF contrasts can produce substantial overestimation when the corresponding weights $w_{kj}$ are non-negligible. Because the row-sums total $p$, this overestimation must be offset by underestimation of other predictors. This structural imbalance explains why severe and heterogeneous multicollinearity can degrade GCD in both quantitative accuracy and ranking consistency. This mechanism also helps explain earlier empirical findings, such as those of~\cite{johnson2000heuristic}, in which GCD was found to underperform. These results identify the lack of row stochasticity in RegPA as the structural source of GCD's a priori bias. 
\subsection{Compound Symmetry as an Analytical Benchmark}
\label{subsec:compound-symmetry}

The preceding results show how VIF contrasts can induce row-sum imbalance in RegPA. To distinguish this mechanism from multicollinearity itself, we now consider compound symmetry, under which all predictors have identical correlation structures and identical VIFs. This setting provides a tractable benchmark for comparing the reallocation behavior of GDA, RegPA, and CorPA. For this analytical comparison, we return to the standardized
design-matrix notation of Section~\ref{sec: RI review}, with
$\Sigma=X^\top X$.

\begin{proposition}[Reallocation matrices under compound symmetry]
\label{prop:compound-symmetry}
Let $p\geq3$ and suppose that
\begin{equation}
    \Sigma
    =
    (1-\rho)I_p+\rho J_p,
    \qquad
    -\frac{1}{p-1}<\rho<1,
    \label{eq:compound-symmetry}
\end{equation}
where $1_p$ is the $p$-vector of ones and $J_p=1_p1_p^\top$. Under Johnson's minimal transformation $Z=X\Sigma^{-1/2}$, define
\[
    \tau_\rho
    =
    \frac{
        \sqrt{1+(p-1)\rho}-\sqrt{1-\rho}
    }{p},
    \qquad
    h_\rho=1+(p-2)\rho,
\]
and
\[
    \kappa_\rho
    =
    \frac{1}{p-1}
    \sum_{s=0}^{p-2}
    \frac{1}{1+s\rho}.
\]
Then
\begin{align}
    \mathrm{CorPA}
    &=
    (1-p\alpha_C)I_p+\alpha_CJ_p,
    &
    \alpha_C
    &=
    \tau_\rho^2,
    \nonumber\\
    \mathrm{GDA}
    &=
    (1-p\alpha_G)I_p+\alpha_GJ_p,
    &
    \alpha_G
    &=
    \kappa_\rho\tau_\rho^2,
    \label{eq:cs-reallocation-matrices}\\
    \mathrm{RegPA}
    &=
    (1-p\alpha_R)I_p+\alpha_RJ_p,
    &
    \alpha_R
    &=
    {\tau_\rho^2}/{h_\rho}.
    \nonumber
\end{align}
Thus, each matrix has diagonal entries $1-(p-1)\alpha_A$ and off-diagonal entries $\alpha_A$, where $A\in\{C,G,R\}$. In particular, for every $i\neq j$,
\begin{equation}
    \mathrm{GDA}_{ij}
    =
    \kappa_\rho\,\mathrm{CorPA}_{ij}
    =
    h_\rho\kappa_\rho\,\mathrm{RegPA}_{ij}.
    \label{eq:cs-off-diagonal-relations}
\end{equation}
When $\rho=0$, all three matrices reduce to $I_p$.
\end{proposition}

The proof of Proposition~\ref{prop:compound-symmetry} is provided in Appendix~\ref{app:proof-compound-symmetry}. The first consequence concerns the a priori bias.

\begin{corollary}[Absence of the a priori bias]
\label{cor:cs-no-bias}
Under the conditions of Proposition~\ref{prop:compound-symmetry}, CorPA, GDA, and RegPA are doubly stochastic. Consequently, for $A\in\{\mathrm{CorPA},\mathrm{GDA},\mathrm{RegPA}\}$,
\begin{equation}
    \mathbb{E}_{\beta_0\sim F}
    \left[
        D_{A,i}(\Sigma,\beta_0)
    \right]
    =
    \frac{1}{p},
    \qquad i=1,\ldots,p.
    \label{eq:cs-expected-importance}
\end{equation}
Moreover, all predictors have the same VIF:
\begin{equation}
    \mathrm{VIF}_i
    =
    \frac{
        1+(p-2)\rho
    }{
        [1+(p-1)\rho](1-\rho)
    },
    \qquad i=1,\ldots,p.
    \label{eq:cs-vif}
\end{equation}
\end{corollary}

The proof is provided in
Appendix~\ref{app:proof-cs-no-bias}. Although the common VIF in Eq.~\eqref{eq:cs-vif} diverges as $\rho\to1^-$, RegPA remains doubly stochastic. Therefore, severe homogeneous multicollinearity is not sufficient to induce GCD's a priori bias. The bias is caused by row-sum imbalance, which does not arise under the exchangeable structure of compound symmetry.

Absence of the a priori bias does not imply that the three reallocation methods agree. Their differences can be characterized by how strongly they shrink importance toward equality. Let
\[
    \eta_j
    =
    \frac{\rho_{z_jy}^2}{R^2_{y\cdot X}},
    \qquad
    \eta=(\eta_1,\ldots,\eta_p)^\top,
    \qquad
    \mu_p=\frac{1}{p}1_p,
\]
so that $1_p^\top\eta=1$.

\begin{corollary}[Shrinkage toward equal importance under compound symmetry]
\label{cor:cs-shrinkage}
Let $D_C$, $D_G$, and $D_R$ denote the normalized ORM importance vectors obtained with the CorPA, GDA, and RegPA reallocations, respectively. For each $M\in\{C,G,R\}$,
\begin{equation}
    D_M-\mu_p
    =
    (1-p\alpha_M)(\eta-\mu_p),
    \qquad
    0<1-p\alpha_M<1
    \quad\text{for }\rho\neq0.
    \label{eq:cs-shrinkage-representation}
\end{equation}
If $0<\rho<1$, then
\begin{equation}
    \alpha_R<\alpha_G<\alpha_C,
    \qquad\text{equivalently}\qquad
    1-p\alpha_C<1-p\alpha_G<1-p\alpha_R,
    \label{eq:cs-alpha-ordering}
\end{equation}
so that, whenever $\eta\neq\mu_p$, for any norm,
\begin{equation}
    \lVert D_C-\mu_p\rVert
    <
    \lVert D_G-\mu_p\rVert
    <
    \lVert D_R-\mu_p\rVert.
    \label{eq:cs-shrinkage-ordering}
\end{equation}
Furthermore,
\begin{equation}
    D_{C,i}-D_{G,i}
    =
    p(\alpha_C-\alpha_G)
    \left(
        \frac{1}{p}-\eta_i
    \right),
    \qquad i=1,\ldots,p,
    \label{eq:cs-correlation-excess-shrinkage}
\end{equation}
with $\alpha_C-\alpha_G>0$. For $-\tfrac{1}{p-1}<\rho<0$, the
orderings in Eqs.~\eqref{eq:cs-alpha-ordering}
and~\eqref{eq:cs-shrinkage-ordering} and the sign of
$\alpha_C-\alpha_G$ are all reversed.
\end{corollary}
 
The proof is provided in Appendix~\ref{app:proof-cs-shrinkage}. Equation~\eqref{eq:cs-shrinkage-representation} shows that, under compound symmetry, every method returns a vector of the same form: the deviation of $\eta$ from equal importance, scaled by a \emph{retention factor} $1-p\alpha_M$. Differences among the three methods therefore reduce to how much of this deviation each one retains. For $0<\rho<1$, the benchmark GDA retains the fraction $1-p\alpha_G$, which represents the degree of shrinkage warranted by the GD criterion itself. CorPA retains strictly less: by Eq.~\eqref{eq:cs-correlation-excess-shrinkage}, relative to GDA it pulls every above-average component of $\eta$ down and every below-average component up, with the size of the distortion governed by the positive excess $p(\alpha_C-\alpha_G)$. Since CorPA combined with Johnson's minimal transformation is exactly RW, this excess shrinkage is a closed-form expression of RW's leveling problem: under compound symmetry, the leveling problem is precisely shrinkage beyond the benchmark, and CorPA is the only method that exhibits it. RegPA deviates in the opposite direction ($1-p\alpha_R>1-p\alpha_G$): it preserves more of the contrasts in $\eta$ than the benchmark deems warranted, an over-concentration rather than a leveling. As a numerical illustration, at $p=3$ and $\rho=0.9$ the retention factors are $0.39$ for CorPA, $0.53$ for GDA, and $0.68$ for RegPA.
 
One caveat applies to ordinal comparisons. Because Eq.~\eqref{eq:cs-shrinkage-representation} rescales $\eta-\mu_p$ by a common positive factor, the three methods produce identical predictor rankings under compound symmetry, and any rank-based performance metric would not distinguish them. The corollary therefore characterizes distortions in the magnitudes of importance values; differences in ordinal accuracy among the methods, such as those observed in Section~\ref{sec: simulations}, can only arise under asymmetric predictor correlation structures.

Together, the two corollaries separate two distinct requirements for a reallocation matrix. Unit row-sums ensure that no predictor's importance is systematically over- or under-estimated in expectation, while agreement with GDA further depends on how the elements are distributed within the doubly stochastic matrix. The present study addresses the first requirement: we next correct the row-sum imbalance of RegPA by mapping it to a nearby doubly stochastic matrix.
\subsection{Bias Correction via Double Stochasticity}
\label{subsec:bias-correction}
Propositions~\ref{prop:expected-importance} and~\ref{prop:gcd-vif-bias} show that GCD's a priori bias arises because the RegPA reallocation matrix is column-stochastic but not generally row-stochastic. We therefore correct RegPA by mapping it to a nearby doubly stochastic matrix, thereby enforcing unit row-sums while preserving the complete allocation of each orthogonal predictor's importance.

Let $\Delta=\{a\in\mathbb{R}^p:1_p^\top a=1, a\geq0\}$ be the probability simplex, $M_r=\{M\in\mathbb{R}_+^{p\times p}:M1_p=1_p\}$ be the set of row-stochastic matrices, and $M_c=\{M\in\mathbb{R}_+^{p\times p}:1_p^\top M=1_p^\top\}$ be the set of column-stochastic matrices. We present two algorithms for constructing a corrected RegPA matrix that lies in $M_r\cap M_c$.
\paragraph{Method of Alternating Projections:} The MAP procedure alternates Euclidean projections onto the $M_r$ and $M_c$ sets. The resulting sequence converges to a matrix in their intersection. At each step, rows and columns are projected onto the probability simplex $\Delta$ by solving:
\begin{equation}
    \min_{\tilde{a}\in\Delta}\frac{1}{2}\|\tilde{a}-a\|^2_2,
    \label{eq: opt. project onto prob. simplex}
\end{equation}
where $a$ is a row or column vector of the matrix. An efficient algorithm to solve this projection problem is provided by~\cite{wang2013projection} (Algorithm~\ref{alg:proj onto prob simplex}). This iterative procedure can be interpreted as alternating projections onto two convex, closed sets $M_r$ and $M_c$ (Algorithm~\ref{alg:MAP}). The von Neumann alternating projection theorem~\citep{von1949rings} guarantees convergence to a point in the intersection $M_r\cap M_c$ (if nonempty). Moreover, due to the linear regularity of $M_r$ and $M_c$~\citep[Corollary 5.26]{bauschke1996projection}, the convergence is linear, as shown in~\citep[Corollary 3.14]{bauschke1993convergence}. We denote the resulting doubly stochastic matrix by RegPA-MAP, and the corresponding corrected ORM as GCD-MAP.

\begin{algorithm}[htb]
    \begin{algorithmic}
        \Require $a\in\mathbb{R}^p$
        \State Sort $a$ into $\overline{a}:\overline{a}_1\geq \overline{a}_2\geq\dots\geq \overline{a}_p$
        \State Find $b=\max\{1\leq j\leq p:\overline{a}_j+\frac{1}{j}(1-\sum_{i=1}^j \overline{a}_i)>0\}$
        \State Define $c=\frac{1}{b}(1-\sum_{i=1}^b\overline{a}_i)$
        \Ensure $\tilde{a}$ s.t. $\tilde{a}_i=\max\{a_i+c, 0\},\,\forall i=1,\dots, p$.
        \caption{Euclidean Projection of a Vector onto the Probability Simplex~\citep{wang2013projection}}
        \label{alg:proj onto prob simplex}
    \end{algorithmic}
\end{algorithm}

\begin{algorithm}[htb]
    \begin{algorithmic}
    \Require A reallocation matrix $A\in\mathbb{R}^{p\times p}$ and a number of iterations $T$
    \State $t=0, A^{(t)}=A$
    \Repeat
        \State $t\leftarrow t+1$
        \State $A^{(t)}=A^{(t-1)}$
        \State Apply Algorithm~\ref{alg:proj onto prob simplex} to each column of $A^{(t)}$.
        \State Apply Algorithm~\ref{alg:proj onto prob simplex} to each row of $A^{(t)}$.
    \Until $t=T$
    \Ensure $A^{(t)}$
    \caption{Method of Alternating Projections to Double Stochasticity}
    \label{alg:MAP}
    \end{algorithmic}
\end{algorithm}

\paragraph{Sinkhorn--Knopp Algorithm (SK):} The Sinkhorn–Knopp (SK) algorithm~\citep{sinkhorn1967concerning} offers another approach to double stochasticity by alternating row and column normalization. Unlike MAP, which projects in Euclidean space, SK performs multiplicative normalization as shown in Algorithm \ref{alg:SK}. Mathematically, applying the SK algorithm to a non-negative matrix $A$ is equivalent to solving the following problem:
\begin{equation}
    \min_{M\in M_r\cap M_c} \mathrm{KL}(M\|A),
    \label{eq:KL div}
\end{equation}
where $\mathrm{KL}(M||A)=\sum_{i=1}^p\sum_{j=1}^p M_{ij}\log\frac{M_{ij}}{A_{ij}}+A_{ij}-M_{ij}$ is the Kullback--Leibler divergence. This is also a Bregman projection onto the doubly stochastic set using KL divergence~\citep{wang2010learning}. Sinkhorn and Knopp have proved that if the matrix is square and strictly positive, then SK is guaranteed to converge to a unique doubly stochastic matrix. The strict positivity also ensures the convergence is linear~\citep{knight2008sinkhorn}. For non-negative matrices (i.e., some zeros allowed), convergence still holds under mild technical conditions. In our problem, RegPA is typically strictly positive, so that the convergence is rapid and reliable. We denote the resulting corrected matrix by RegPA-SK, and the corresponding ORM as GCD-SK. 


\begin{algorithm}[htb]
    \begin{algorithmic}
    \Require A reallocation matrix $A\in\mathbb{R}^{p\times p}$ and a number of iterations $T$
    \State $t=0, A^{(t)}=A$
    \Repeat
        \State $t\leftarrow t+1$
        \State $A^{(t)}=A^{(t-1)}$
        \State Make col-sums equal to 1 by normalizing $\tilde{a}_{ij}^{(t)}={a_{ij}^{(t)}}/{\sum_{k=1}^p a_{kj}^{(t)}},\,\forall i,j=1,\dots,p.$
        \State Make row-sums equal to 1 by normalizing $a_{ij}^{(t)}={\tilde{a}_{ij}^{(t)}}/{\sum_{\ell=1}^p \tilde{a}_{i\ell}^{(t)}},\,\forall i,j=1,\dots,p.$
    \Until $t=T$
    \Ensure $A^{(t)}$
    \caption{Sinkhorn--Knopp Algorithm~\citep{sinkhorn1967concerning}}
    \label{alg:SK}
    \end{algorithmic}
\end{algorithm}
\subsection{Numerical Illustration}
We illustrate the a priori bias and the effects of the proposed corrections using a setting with $p=10$ predictors and severe multicollinearity. The purpose of this example is to visualize how row-sum imbalance in RegPA distorts GCD importance and how the MAP and SK corrections alter both the reallocation matrix and the resulting importance estimates. Figure~\ref{fig:ReA} displays the reallocation matrices for four methods: RegPA, RegPA-MAP, RegPA-SK, and the benchmark GDA. The original RegPA matrix exhibits substantial row-sum imbalance. In particular, the row corresponding to $x_4$ sums to 2.58, meaning that $x_4$ receives more than twice its fair share of total importance on average across possible responses. 

Figure~\ref{fig:RI comparison} shows the resulting RI values produced by different methods. The RI values are ordered by GD (shown in black), which serves as the reference. As expected, the uncorrected GCD (blue) deviates notably from GD, dramatically overestimating the importance of $x_4$. In contrast, the corrected methods, GCD-MAP (light blue) and GCD-SK (light green), closely track the GD rankings, indicating that both corrections successfully mitigate the bias.

Moreover, the structural similarity of the corrected reallocation matrices to the benchmark GDA can be quantified using the Frobenius distance: $    d_F(A,B)=\|A-B\|_F=(\sum_{i=1}^p\sum_{j=1}^p(a_{ij}-b_{ij})^2)^{1/2}.$ The distances are summarized in Table~\ref{tab:frob dis}.
\begin{table}[htb]
    \centering
    \caption{The Frobenius distances of RegPA and its corrections to the benchmark GDA in the numerical illustration.}
    \begin{tabular}{|c|c|c|} \hline
         $d_F(\text{RegPA},\text{GDA})$ & $d_F(\text{RegPA-MAP},\text{GDA})$ & $d_F(\text{RegPA-SK},\text{GDA})$ \\ \hline
         0.83 & 0.36 & 0.49 \\ \hline 
    \end{tabular}
    \label{tab:frob dis}
\end{table}
These results illustrate that both MAP and SK substantially reduce the discrepancy from GDA, helping to explain the improved RI performance.

\begin{figure}[tb]
     \centering
     \begin{subfigure}[b]{0.48\textwidth}
         \centering
         \includegraphics[width=\textwidth]{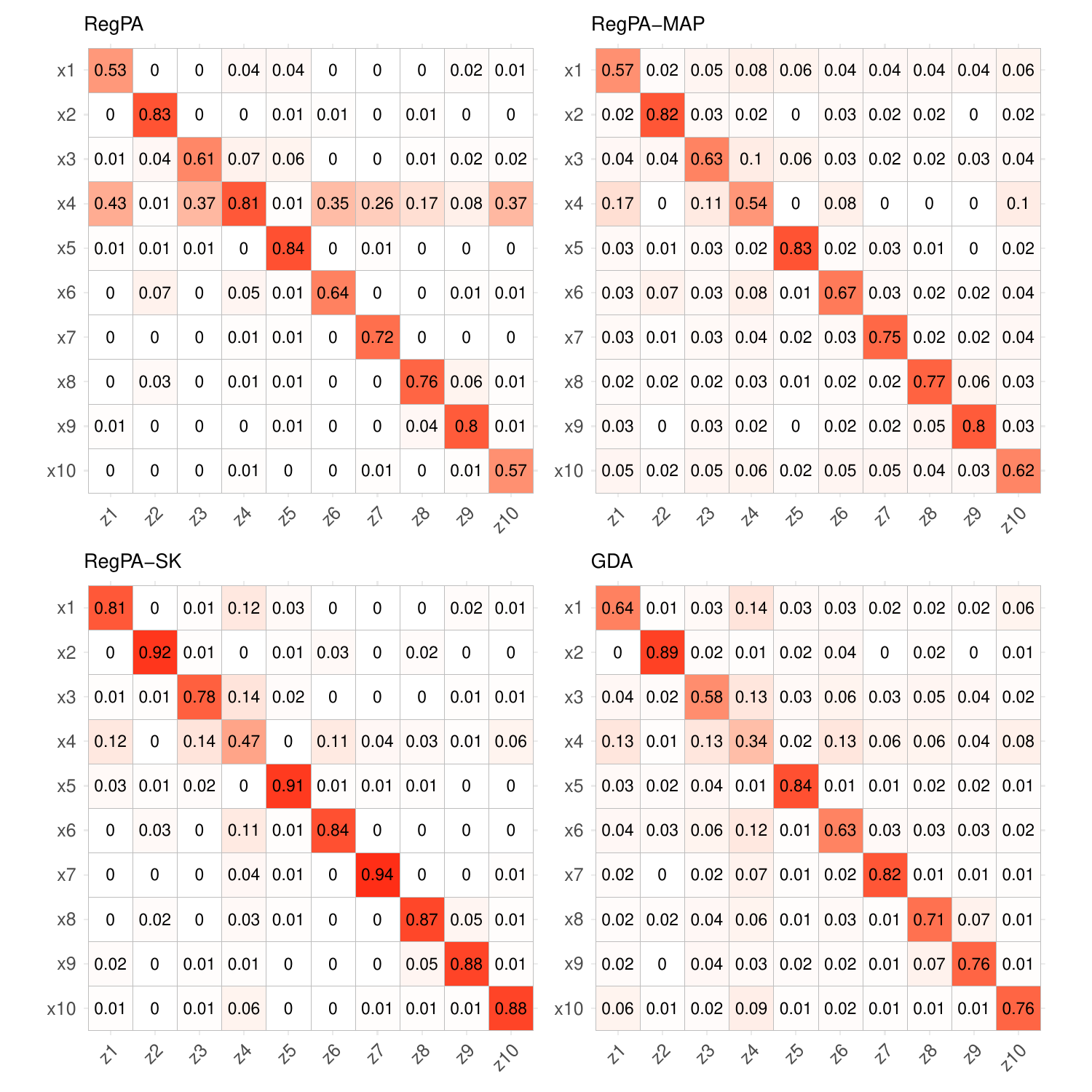}
         \caption{Reallocation matrices: RegPA (top left) exhibits severe row-sum imbalance; corrected matrices via MAP and SK (top right and bottom left) are doubly stochastic and closer to the benchmark GDA (bottom right).}
         \label{fig:ReA}
     \end{subfigure}
     \begin{subfigure}[b]{0.48\textwidth}
         \centering
         \includegraphics[width=\textwidth]{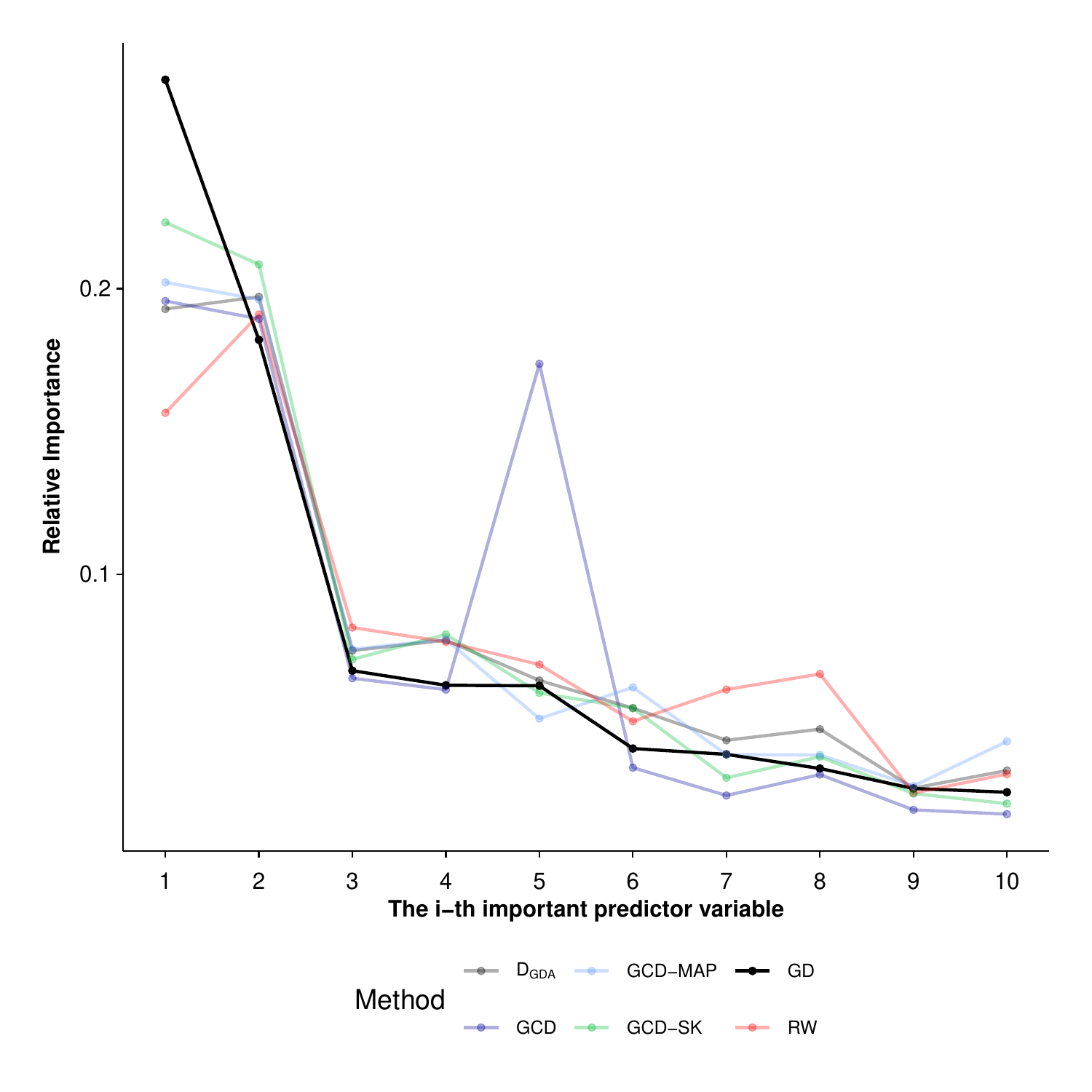}
         \caption{Relative importance values computed by different methods. Predictors are ordered by GD importance (black). GCD (blue) significantly overestimates $x_4$ due to the a priori bias, while corrected versions GCD-MAP (light blue) and GCD-SK (light green) more closely follow the GD trend.}
         \label{fig:RI comparison}
     \end{subfigure}
        \caption{A $p=10$ numerical illustration of the a priori bias in RegPA (GCD) and its correction via MAP and SK algorithms.}
        \label{fig:demo}
\end{figure}
\section{Simulations}\label{sec: simulations}
To evaluate the effectiveness of our proposed corrections for the a priori bias, we adopt the comprehensive Monte Carlo simulation framework developed by~\cite{chang2025understanding}. Their design is specifically aimed at assessing how well ORMs approximate the gold standard GD across a wide range of predictor correlation structures and response scenarios.
\subsection{Simulation Design}
The evaluation centers on estimating the expected (dis)similarity between an ORM and GD over all possible responses, for any fixed predictor correlation matrix $\Sigma$. Formally, the expected performance evaluation $f$ is defined as: 
\begin{equation}
    f(\Sigma)=\mathbb{E}_{\beta}\left[g\left(D_{A}(\Sigma,\beta), D(\Sigma,\beta)\right)\right],
    \label{eq: expected metric}
\end{equation}
where $D_A$ is the ORM using minimal transformation $Z$ for orthogonalization and $A$ for reallocation, and $D$ denotes the GD vector. The function $g(\cdot,\cdot)$ is a (dis)similarity metric. We use (1) Root Mean Square Error (RMSE) to measure quantitative deviation (lower is better) and (2) Kendall's $\tau$ to assess ordinal consistency (higher is better)~\citep{kendall1938new}. This expectation is approximated via Monte Carlo simulation with $B$ samples:
\begin{equation}
    \hat{f}(\Sigma)=\frac{1}{B}\sum_{b=1}^{B}\left[g\left(D_{A}(\Sigma,\beta^{(b)}), D(\Sigma,\beta^{(b)})\right)\right],
    \label{eq: realized metric}
\end{equation}
where each ${\beta^{(b)}}$ is drawn from $F$. This setup mirrors real-world applications (e.g., behavioral sciences), where multiple responses are analyzed over the same set of predictors. Evaluating ORMs across all possible responses provides a principled way to assess generalizability with respect to the predictor structure.
\subsection{Simulation Procedure}
To generate diverse and representative predictor correlation structures, the simulation consists of the following steps:
\begin{enumerate}
    \item Varying number of predictors. \\ The number of predictors $p$ is varied from 3 to 10.
    
    \item Sampling eigenvalues. \\ For each $p$, we generate $n_{\text{ev}}$ eigenvalue sets sampled uniformly from the $(p-1)$-dimensional simplex $\Delta^{p-1}=\{\lambda\in\mathbb{R}^p:\sum_{i=1}^p \lambda_i=p, \lambda_1\geq\ldots\geq \lambda_p>0\}$ using Kraemer's algorithm~\citep{smith2004sampling}.
    
    \item Generating predictor correlation matrices.\\ For each eigenvalue set, $n_s$ distinct correlation matrices $\Sigma$ are generated using the \texttt{rMAP} function in the \texttt{R} package \texttt{fungible}~\citep{waller2020generating}. These matrices share the same eigenvalues but differ in their eigenvectors, which significantly affect multicollinearity (e.g., VIFs).
    
    \item Sampling responses. \\ For each predictor correlation matrix $\Sigma$
    \begin{itemize}
        \item Sample $n_u$ coefficient vectors $u\sim\mathcal{U}(S^{p-1})$.
        \item Transform each $u$ to $\beta=\Sigma^{-1/2}u$ and compute $\rho_{xy}=\Sigma\beta=\Sigma^{1/2}u$.
    \end{itemize}
    
    \item Computing importance measures. \\ For each generated correlation structure $\Sigma$ and $\rho_{xy}$, compute: GD, ORMs with various reallocation methods (IdA, RegPA, CorPA, GDA), and ORMs corrected via proposed MAP and SK strategies (RegPA-MAP, RegPA-SK).
    
    \item Evaluating performance. \\ Each ORM is evaluated against GD using both RMSE and Kendall's $\tau$. The Monte Carlo approximation $\hat{f}(\Sigma)$ in Eq.~\eqref{eq: realized metric} is computed for each correlation matrix.
\end{enumerate}
The parameters are set as follows: number of eigenvalue sets $n_{\text{ev}}=1000$ for $p\in\{3,4,5,6\}$ and $n_{\text{ev}}=2500$ for $p\in\{7,8,9,10\}$; number of predictor correlation matrices for a given eigenvalue set $n_s=10$; number of responses evaluated per correlation matrix $n_u=100$ (equivalently, $B=100$). The number of iterations was set to $T=100$ for both MAP and SK, which was empirically sufficient to achieve numerical convergence in all simulated settings. Code to reproduce the simulation study and empirical illustration is available via the OSF repository: \url{https://osf.io/zga2p/overview?view_only=6b7f4b27c3924886a35b281550efacc6}.
\subsection{Results}
To facilitate interpretation, we follow the evaluation framework of~\cite{chang2025understanding}, which partitions the comprehensive Monte Carlo simulation results into four mutually exclusive scenarios, defined by two key factors of the predictor correlation structure together with their corresponding thresholds: (1) $\lambda_1/\sqrt{p}$, where $\lambda_1$ is the largest eigenvalue of the predictor correlation matrix, reflecting the dominance of the first principal component, and (2) $\mathrm{VIF}_{\max}/p$, capturing the severity of multicollinearity. For each predictor $x_i$, VIF is computed as $\mathrm{VIF}_i=1/(1-R^2_{x_i\cdot X_{-i}})$ where $X_{-i}$ is the matrix of all predictors except $x_i$ and the maximum is taken across all predictors.

Using the thresholds $\lambda_1/\sqrt{p}=1.5$ and $\mathrm{VIF}_{\max}/p=4$ adopted in this framework, the simulation results are categorized into four scenarios based on combinations of mild/severe multicollinearity and mild/strong first principal component. Figures~\ref{fig:s1.1 and s1.2} and~\ref{fig:s2.1 and s2.2} present the performance of each reallocation method in terms of RMSE and Kendall's $\tau$. Since the orthogonalization method is fixed to the minimal transformation $Z$, we refer to each ORM by its reallocation matrix and use the names interchangeably. Each point corresponds to a Monte Carlo estimate of $f(\Sigma)$ under a distinct predictor correlation matrix. 
\subsubsection{Scenario 1.1: Mild Multicollinearity with Mild First Principal Component}
\begin{figure}[tb]
    \centering
    \includegraphics[width=\linewidth]{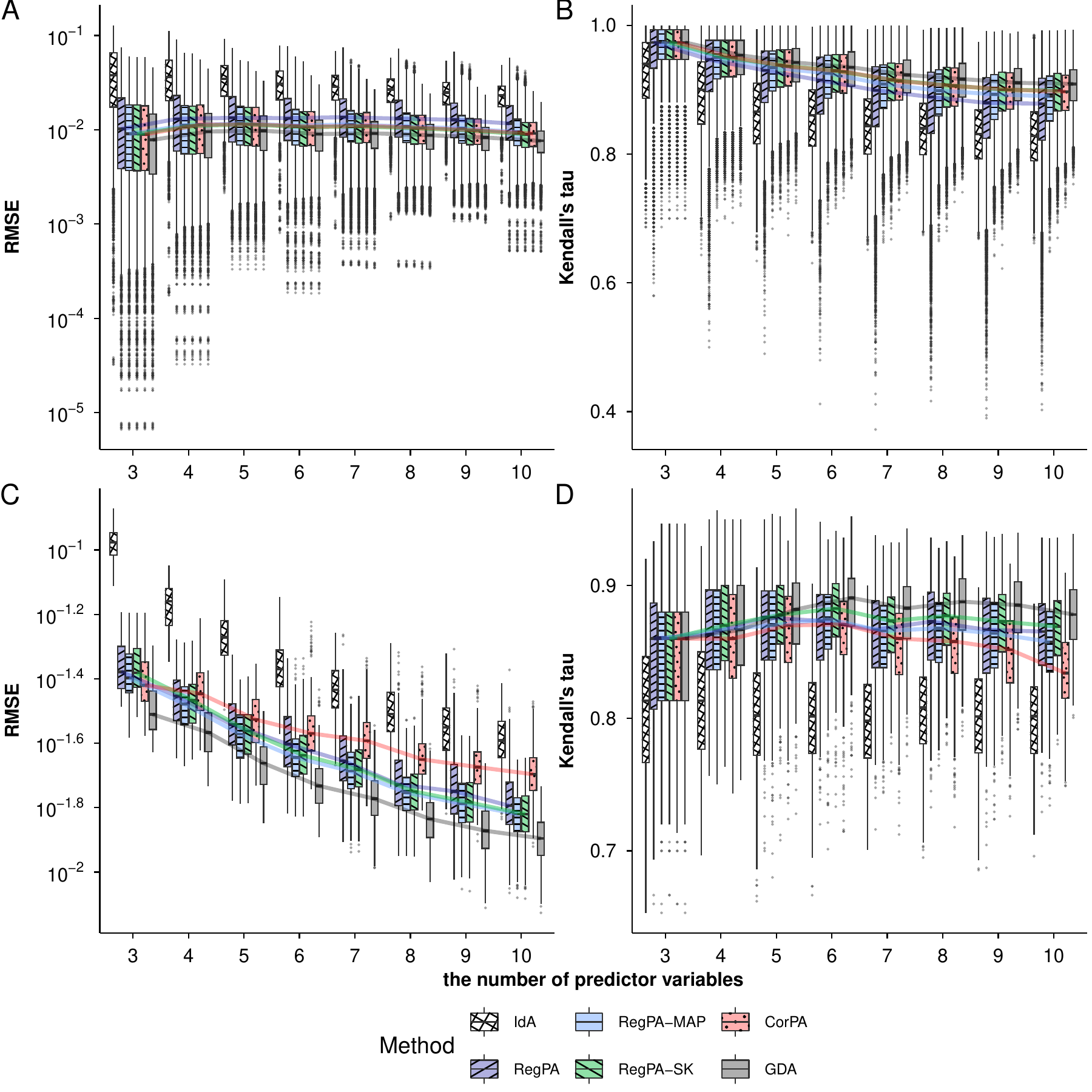}
    \caption{Performance of ORMs across predictor structures with mild multicollinearity and mild first principal component (A, B) and strong first principal component (C, D). Median performances of RegPA, RegPA-MAP, RegPA-SK, CorPA and GDA are highlighted in blue, light blue, light green, red and gray lines, respectively.}
    \label{fig:s1.1 and s1.2}
\end{figure}
In Scenario 1.1 (Figure~\ref{fig:s1.1 and s1.2}A, B), all methods perform comparably and closely approximate GD, with IdA performs consistently worse across both metrics. For this scenario, the a priori bias in GCD (RegPA) is minimal. Accordingly, the proposed corrections applied to RegPA yield only marginal improvements. As expected, the benchmark GDA achieves the strongest overall performance, with ORMs closely following. These results indicate that GCD-SK (RegPA-SK) is a viable and robust alternative to RW when multicollinearity is mild and predictor structure is not dominated by the first principal component.
\subsubsection{Scenario 1.2: Mild Multicollinearity with Strong First Principal Component}
In Scenario 1.2 (Figure~\ref{fig:s1.1 and s1.2}C, D), the first principal component is dominant, leading RW to suffer from the leveling problem, which tends to flatten the predictor importance. This flattening is consistent with the excess shrinkage of CorPA established in Corollary~\ref{cor:cs-shrinkage} for the symmetric case. As previously observed in~\cite{chang2025understanding}, this results in noticeably degraded performance of RW.

In contrast, GCD performs better in this setting, and its corrected versions (GCD-MAP and GCD-SK) further improve upon this performance, approaching the benchmark GDA. Of the two, GCD-SK again achieves better consistency with GD in ranking the variables' importance. These findings demonstrate both the underlying strength of RegPA in such predictor correlation structures and the efficacy of our proposed corrections.

\begin{figure}[tb]
    \centering
    \includegraphics[width=\linewidth]{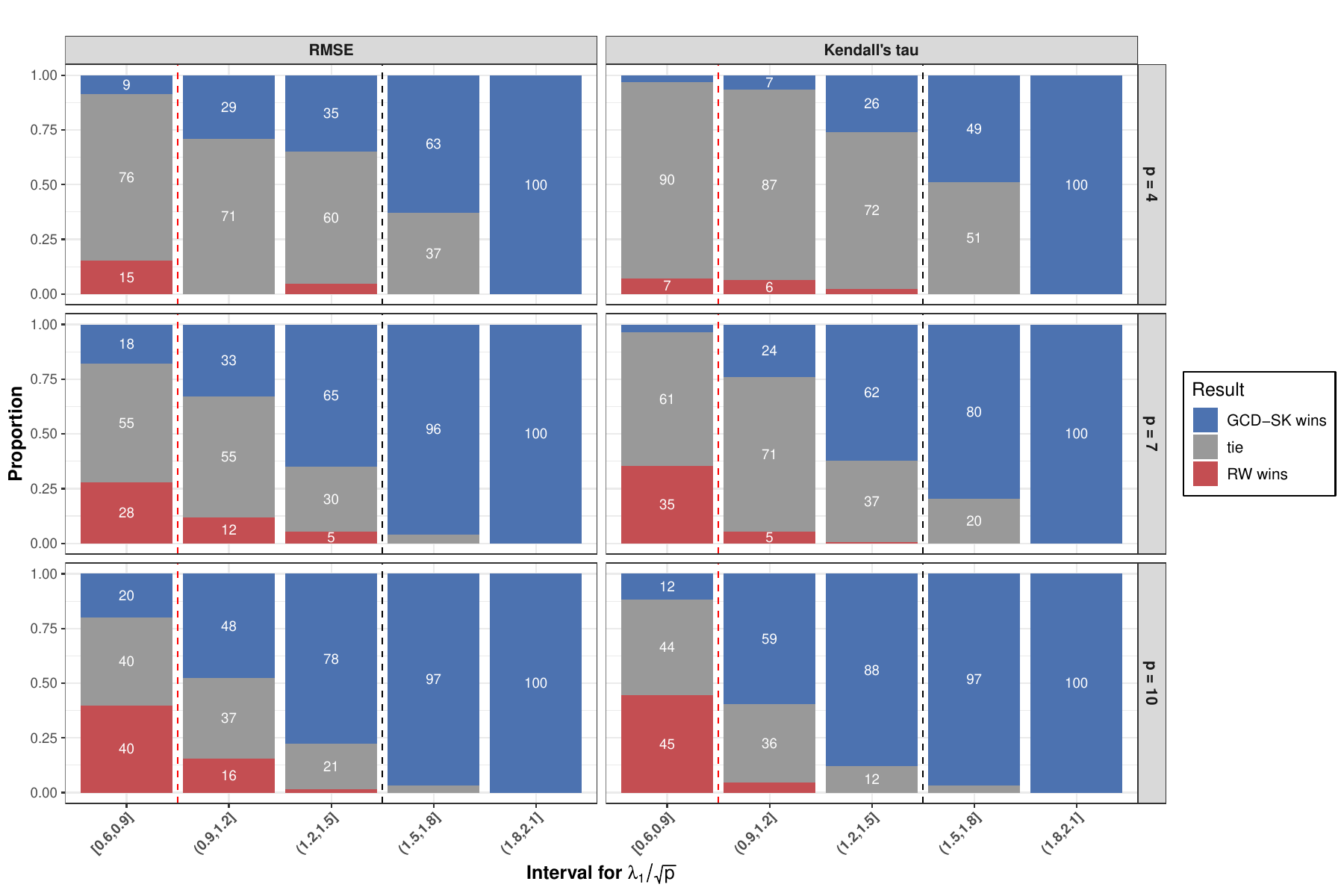}
    \caption{Win--loss analysis between RW and GCD-SK under mild     multicollinearity (Scenarios 1.1 and 1.2), based on RMSE and Kendall's $\tau$. Each comparison corresponds to one eigenvalue set: a method wins if a Wilcoxon signed-rank test over the $n_s=10$ paired correlation matrices sharing those eigenvalues indicates significantly better performance at the 5\% level, and the comparison is a tie otherwise. The dashed line marks the threshold $\lambda_1/\sqrt{p}=1.5$, beyond which the first principal component is considered dominant.}
    \label{fig:s1 win-loss}
\end{figure}
To further quantify these differences, we conduct a win--loss analysis between RW and GCD-SK (Figure~\ref{fig:s1 win-loss}), adapting the approach of~\cite{chang2025understanding}. The analysis exploits the structure of the simulation design: each eigenvalue set is associated with $n_s=10$ correlation matrices that share the same eigenvalues---and hence the same $\lambda_1/\sqrt{p}$---but differ in their eigenvectors. For each eigenvalue set, the Monte Carlo estimates $\hat{f}(\Sigma)$ of the two methods on these ten matrices form ten matched pairs, to which we apply a two-sided Wilcoxon signed-rank test at the 5\% significance level, separately for RMSE and Kendall's $\tau$. If the test is significant, the eigenvalue set is recorded as a win for the method with the better paired performance (lower RMSE or higher $\tau$); otherwise it is recorded as a tie. Figure~\ref{fig:s1 win-loss} reports the resulting proportions of wins and ties as a function of $\lambda_1/\sqrt{p}$. As $\lambda_1/\sqrt{p}$ increases, RW's win proportion drops sharply, reflecting the growing severity of the leveling problem. The comparison with uncorrected GCD quantifies the gain from the SK correction: in the corresponding RW-versus-GCD analysis of~\cite{chang2025understanding}, the crossover beyond which the regression-based measure prevails occurs at $\lambda_1/\sqrt{p}=1.5$, whereas the win proportion of GCD-SK exceeds that of RW for $\lambda_1/\sqrt{p}>0.9$. The SK correction thus substantially enlarges the region of predictor structures in which the regression-based reallocation is preferable, extending it to structures with only moderate first principal component dominance.
\subsubsection{Scenario 2.1: Severe Multicollinearity with Mild First Principal Component}
\begin{figure}[tb]
    \centering
    \includegraphics[width=\linewidth]{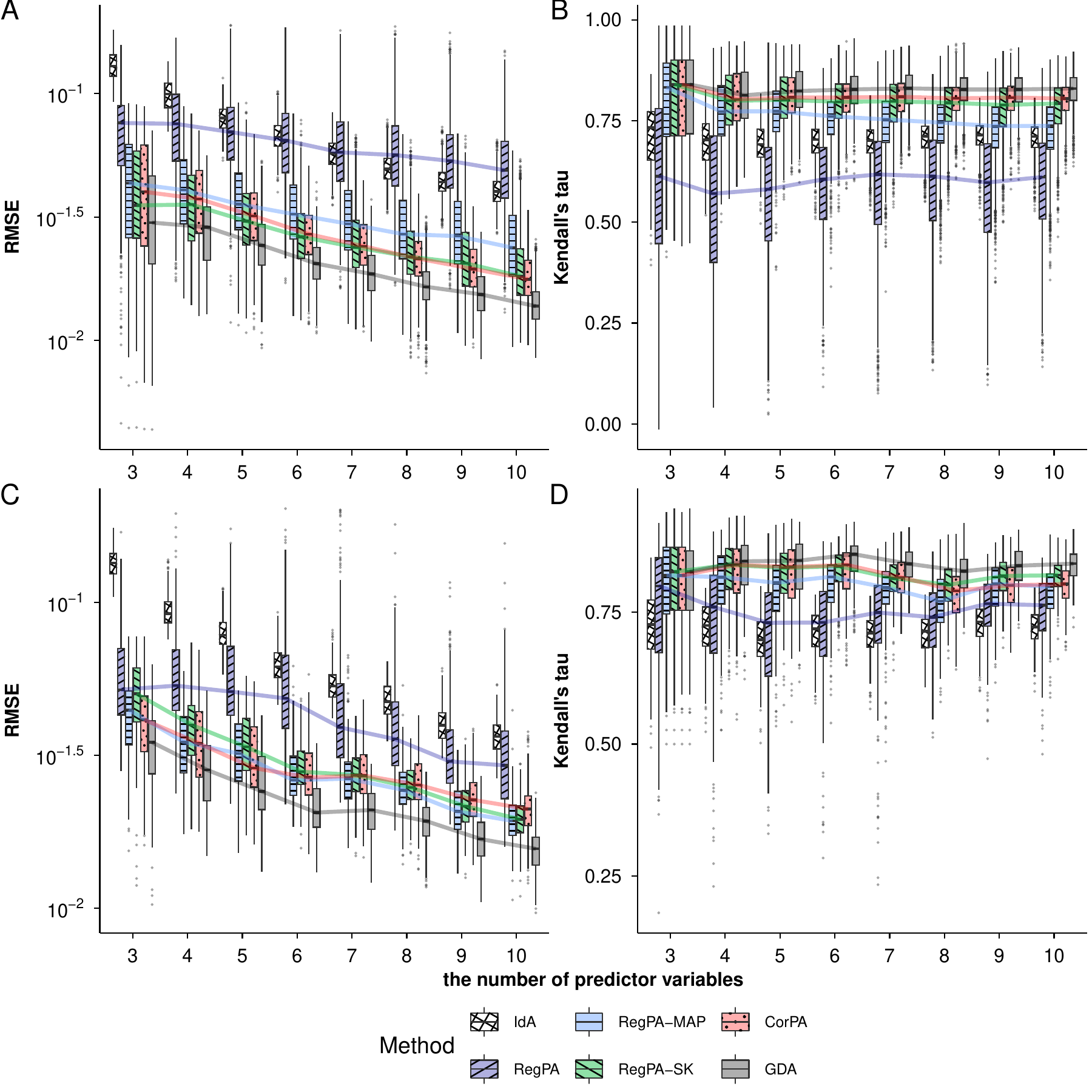}
    \caption{Performance of ORMs across predictor structures with severe multicollinearity and mild first principal component (A, B) and strong first principal component (C, D). Median performances of RegPA, RegPA-MAP, RegPA-SK, CorPA and GDA are highlighted in blue, light blue, light green, red and gray lines, respectively.}
    \label{fig:s2.1 and s2.2}
\end{figure}
Scenario 2.1 (Figure~\ref{fig:s2.1 and s2.2}A, B) presents severe multicollinearity without dominance of the first principal component. In this setting, the a priori bias in RegPA becomes pronounced, leading to degraded GCD performance. However, both corrections effectively eliminate this bias, restoring performance to a level comparable to RW. Between the two, GCD-SK shows superior performance in both RMSE and Kendall's $\tau$.
\subsubsection{Scenario 2.2: Severe Multicollinearity with Strong First Principal Component}
Scenario 2.2 (Figure~\ref{fig:s2.1 and s2.2}C, D) represents the most difficult case, combining severe multicollinearity and a dominant principal component. As expected, both RW and GCD perform poorly. Nevertheless, the corrected methods, especially GCD-SK, recover some of the lost performance. GCD-SK consistently outperforms RW as the number of predictors increases. This illustrates that correcting the a priori bias can significantly improve the robustness of RegPA even under difficult conditions.

\begin{figure}[tb]
    \centering
    \includegraphics[width=\linewidth]{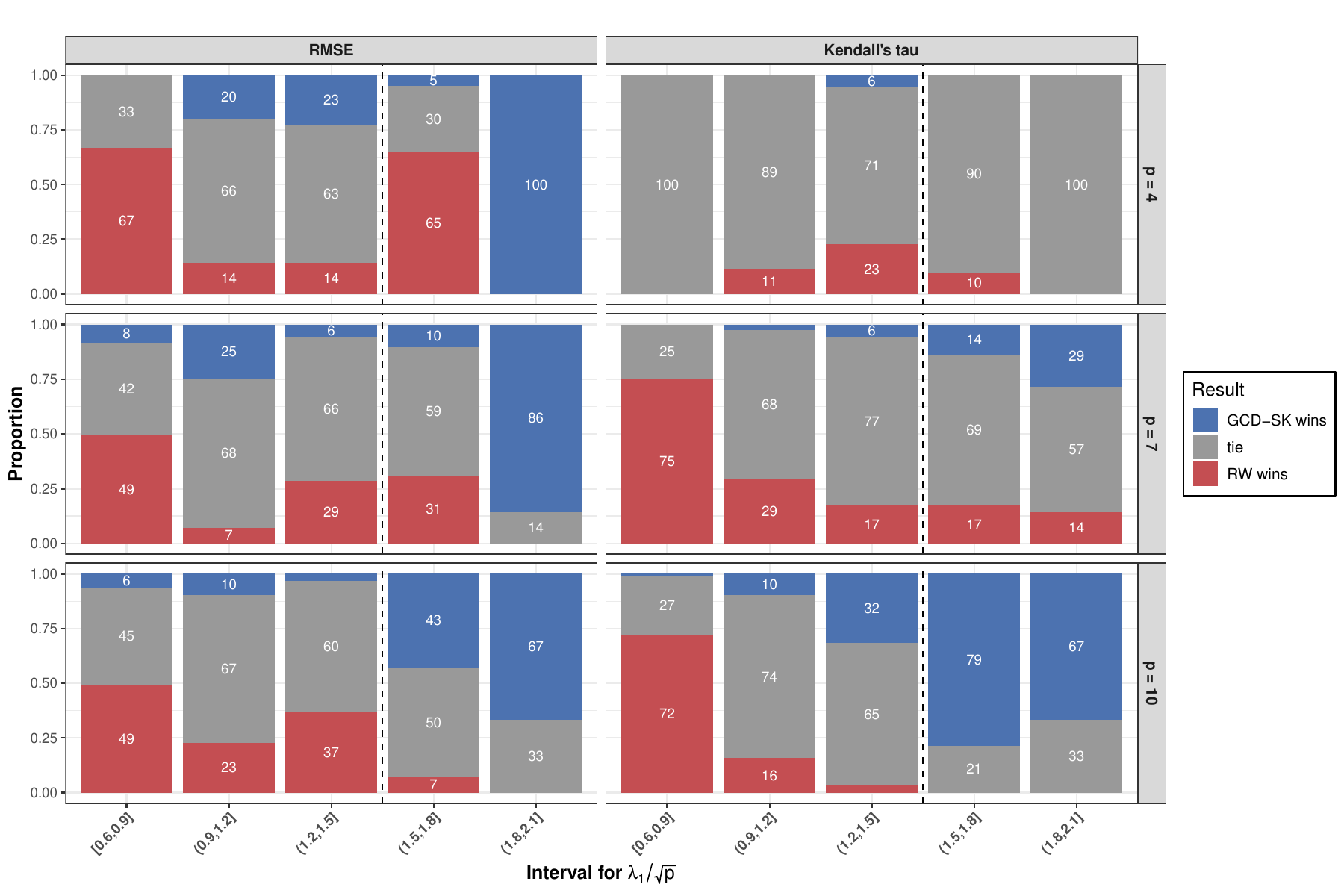}
    \caption{Win--loss analysis between RW and GCD-SK under severe multicollinearity (Scenarios 2.1 and 2.2), based on RMSE and Kendall's $\tau$. Wins and ties are determined as in Figure~\ref{fig:s1 win-loss}.}
    \label{fig:s2 win-loss}
\end{figure}
Figure~\ref{fig:s2 win-loss} presents the win--loss analysis for the severe multicollinearity scenarios, where the improvement over uncorrected GCD is most striking. In this regime, the a priori bias is so pronounced that uncorrected GCD fails to outperform RW at any level of $\lambda_1/\sqrt{p}$~\citep{chang2025understanding}. Removing the bias restores the competitiveness of the regression-based reallocation: GCD-SK overtakes RW once $\lambda_1/\sqrt{p}$ exceeds approximately $1.5$, with consistent gains in both RMSE and Kendall's $\tau$. Taken together, the two crossover points---$\lambda_1/\sqrt{p}\approx0.9$ under mild and $\approx1.5$ under severe multicollinearity---form the empirical basis of the selection guidelines in Section~\ref{subsec:practical-guidelines}.
\subsection{Practical Guidelines}\label{subsec:practical-guidelines}
To avoid the a priori bias, an ORM should use a doubly stochastic reallocation matrix. Among the bias-free methods considered here, the simulation results suggest the following guidelines:
\begin{enumerate}
    \item Mild Multicollinearity ($\mathrm{VIF}_{\max}/p<4$):
    \begin{itemize}
        \item If the first principal component is weak $(\lambda_1/\sqrt{p}<0.9)$, RW is the preferred choice due to its simplicity and high ordinal consistency.
        \item If the first principal component is moderate to strong $(\lambda_1/\sqrt{p}\geq 0.9)$, GCD-SK is recommended. In this range, RW begins to suffer from the leveling problem, while GCD-SK maintains robustness. 
        \item For reference, even the uncorrected GCD outperforms RW when $\lambda_1/\sqrt{p}\geq 1.5$. 
    \end{itemize}
    \item Severe Multicollinearity ($\mathrm{VIF}_{\max}/{p} \geq 4$):
    \begin{itemize}
        \item If the first principal component is weak $(\lambda_1/\sqrt{p} < 1.5)$, RW remains a viable option. In this setting, the a priori bias causes uncorrected GCD to underperform significantly, whereas RW remains robust. While GCD-SK successfully eliminates this bias and drastically improves the uncorrected GCD, empirical results indicate that RW still maintains a slight performance advantage over GCD-SK in this specific scenario.
        \item If the first principal component is strong ($\lambda_1/\sqrt{p} \geq 1.5$), both standard RW and uncorrected GCD struggle. In this challenging regime, GCD-SK is the recommended method, particularly as the number of predictors increases, offering a balanced performance.
    \end{itemize}
These thresholds are empirical rather than universal constants. They summarize the present simulation design for $p=3,\ldots,10$ and should be interpreted as practical diagnostics rather than sharp theoretical boundaries.
\end{enumerate}
\section{Discussion and Conclusion}\label{sec: discussion and conclusion}

Orthogonalization-Reallocation Measures provide computationally efficient approximations to GD, but their performance depends critically on how the importance of orthogonal predictors is reallocated. In this study, we formalized the a priori bias by showing that the expected normalized importance assigned by an ORM is determined by the row-sums of its reallocation matrix. For GCD, we further showed that these row-sums can be expressed as weighted averages of VIF ratios, explaining why heterogeneous predictor structures can systematically favor high-VIF predictors.

The compound-symmetry analysis provides a complementary perspective. Under this homogeneous correlation structure, GDA, RegPA, and CorPA admit closed-form expressions and are all doubly stochastic. Consequently, GCD remains free of the a priori bias even as the common VIF diverges. This demonstrates that severe multicollinearity alone is insufficient to produce the bias; row-sum imbalance is the essential mechanism. The same analysis also shows that, for positive correlations, CorPA applies stronger shrinkage toward equal importance than GDA, providing an analytical explanation of RW's leveling behavior in this setting. Thus, the a priori bias and the leveling problem are structurally distinct: the former concerns unequal expected allocation across predictors, whereas the latter can arise even from a doubly stochastic reallocation matrix.

To correct the a priori bias, we proposed enforcing double stochasticity through the Method of Alternating Projections and the Sinkhorn--Knopp algorithm, yielding GCD-MAP and GCD-SK. The simulation results show that GCD-SK generally provides a closer approximation to GD than GCD-MAP. This suggests that multiplicative scaling may preserve the informative allocation structure of RegPA more effectively than the Euclidean projections used by MAP. Across a broad range of predictor correlation structures, GCD-SK substantially improves upon uncorrected GCD and frequently outperforms RW, particularly when RW is affected by leveling. These comparisons are summarized in the selection guidelines, which recommend RW or GCD-SK according to the dominance of the first principal component and the severity of multicollinearity.

These findings also clarify the value of regression-based reallocation. RegPA uses regression coefficients rather than marginal correlations between the original and orthogonal predictors, and this joint information can be useful once the systematic row-sum imbalance is removed. Although the present compound-symmetry result explains the leveling mechanism in an analytically tractable setting, a general characterization and correction of leveling under arbitrary predictor correlation structures remain open problems. More broadly, understanding why computationally distinct ORMs can closely approximate GD remains an important direction for future research.
\section*{Acknowledgment}
This work was partially supported by Grant NSTC 114-2221-E-002-168-MY2 from the National Science and Technology Council, Taiwan.

\appendix

\section{Technical Derivations and Proofs}
\label{app:technical-derivations}
\subsection{Proof of Proposition~\ref{prop:expected-importance}}
\label{app:proof-expected-importance}

\begin{proof}
Let $P=\{1,\ldots,p\}$, and $B=\Sigma^{1/2},$ so that $x_0=Bz_0$. For any subset $S\subseteq P$, let $B_S$ be the submatrix formed by the rows of $B$ indexed by $S$. Since $B$ is nonsingular, $B_S$ has rank $|S|$. For nonempty $S$, the population coefficient of determination from regressing $y_0$ on $x_{0,S}$ is
\begin{align}
    R^2_{y_0\cdot x_{0,S}}
    &=
    \frac{
        \Cov(y_0,x_{0,S})
        \Var(x_{0,S})^{-1}
        \Cov(x_{0,S},y_0)
    }{\Var(y_0)}
    \nonumber\\
    &=
    \frac{
        u_0^\top B_S^\top
        (B_SB_S^\top)^{-1}
        B_Su_0
    }{1+\sigma^2},
    \label{eq:appendix-subset-r2}
\end{align}
because $\Var(y_0)=\|u_0\|_2^2+\sigma^2=1+\sigma^2$.

Define $\Pi_S=B_S^\top(B_SB_S^\top)^{-1}B_S$. Then $\Pi_S$ is the orthogonal projection onto the row space of $B_S$, and therefore
\[
    \tr(\Pi_S)=\rank(\Pi_S)=|S|.
\]
For $S=\varnothing$, set $\Pi_{\varnothing}=0$ and $R^2_{y_0\cdot x_{0,\varnothing}}=0$. For the full predictor set, $\Pi_P=I_p$, and hence
\[
    R^2_{y_0\cdot x_0}
    =
    \frac{1}{1+\sigma^2}.
\]
The normalized subset value function is consequently
\[
    v_{\beta_0}(S)
    =
    \frac{R^2_{y_0\cdot x_{0,S}}}
         {R^2_{y_0\cdot x_0}}
    =
    u_0^\top\Pi_Su_0.
\]

Rotational invariance of the uniform distribution on $S^{p-1}$ implies
\[
    \mathbb{E}_{u_0}[u_0u_0^\top]
    =
    \frac{1}{p}I_p.
\]
Therefore,
\begin{align}
    \mathbb{E}_{\beta_0\sim F}
    \left[
        v_{\beta_0}(S)
    \right]
    &=
    \tr\left(
        \Pi_S\mathbb{E}_{u_0}[u_0u_0^\top]
    \right)
    \nonumber\\
    &=
    \frac{\tr(\Pi_S)}{p}
    =
    \frac{|S|}{p}.
    \label{eq:appendix-expected-subset-r2}
\end{align}

The normalized GD importance is the Shapley allocation of $v_{\beta_0}$:
\[
    D_i(\Sigma,\beta_0)
    =
    \frac{1}{p}
    \sum_{S\subseteq P\setminus\{i\}}
    \frac{
        v_{\beta_0}(S\cup\{i\})-v_{\beta_0}(S)
    }{
        \binom{p-1}{|S|}
    }.
\]
By Eq.~\eqref{eq:appendix-expected-subset-r2}, every expected marginal contribution equals
\[
    \mathbb{E}_{\beta_0\sim F}
    \left[
        v_{\beta_0}(S\cup\{i\})-v_{\beta_0}(S)
    \right]
    =
    \frac{1}{p}.
\]
The corresponding GD weights sum to one:
\[
    \frac{1}{p}
    \sum_{S\subseteq P\setminus\{i\}}
    \frac{1}{\binom{p-1}{|S|}}
    =
    \frac{1}{p}
    \sum_{s=0}^{p-1}1
    =
    1.
\]
It follows that
\[
    \mathbb{E}_{\beta_0\sim F}
    \left[
        D_i(\Sigma,\beta_0)
    \right]
    =
    \frac{1}{p}.
\]

For the ORM identity, let $\tilde z_0$ be its standardized orthogonal representation. Since $\tilde z_0$ and $z_0$ are orthonormal representations of the same predictor space, there exists an orthogonal matrix $Q\in\mathbb{R}^{p\times p}$ such that $\tilde z_0=Q^\top z_0$. Defining $\tilde u_0=Q^\top u_0$ gives $z_0^\top u_0=\tilde z_0^\top\tilde u_0$. By rotational invariance,
\[
    \tilde u_0\sim\mathcal{U}(S^{p-1}),
    \qquad
    \mathbb{E}_{u_0}[\tilde u_{0j}^2]
    =
    \frac{1}{p}.
\]

Because $\Var(\tilde z_{0j})=1$ and $\Cov(\tilde z_{0j},y_0)=\tilde u_{0j}$,
\[
    \rho_{\tilde z_{0j}y_0}^2
    =
    \frac{\tilde u_{0j}^2}{1+\sigma^2}.
\]
After normalization by $R^2_{y_0\cdot x_0}=1/(1+\sigma^2)$, the ORM importance is
\[
    D_{A,i}(\Sigma,\beta_0)
    =
    \sum_{j=1}^p a_{ij}\tilde u_{0j}^2.
\]
Since $A$ is response-independent,
\begin{align}
    \mathbb{E}_{\beta_0\sim F}
    \left[
        D_{A,i}(\Sigma,\beta_0)
    \right]
    &=
    \sum_{j=1}^p
    a_{ij}
    \mathbb{E}_{u_0}[\tilde u_{0j}^2]
    \nonumber\\
    &=
    \frac{1}{p}\sum_{j=1}^p a_{ij}.
\end{align}
This proves both identities.
\end{proof}
\subsection{Proof of Proposition~\ref{prop:gcd-vif-bias}}
\label{app:proof-gcd-vif-bias}

\begin{proof}
Under Johnson's minimal transformation,
\[
    z_0=\Sigma^{-1/2}x_0,
\]
so the matrix of regression coefficients of $z_0$ on $x_0$ is
\[
    \Gamma_Z
    =
    \Sigma^{-1}\Cov(x_0,z_0)
    =
    \Sigma^{-1}\Sigma^{1/2}
    =
    \Sigma^{-1/2}.
\]
Because $\Gamma_Z$ is symmetric,
\[
    \Gamma_Z^\top\Gamma_Z
    =
    \Gamma_Z\Gamma_Z^\top
    =
    \Sigma^{-1}.
\]
Therefore, both the squared norm of the $j$-th column and the squared norm of the $j$-th row of $\Gamma_Z$ equal
\[
    \sum_{i=1}^p\gamma_{ij}^2
    =
    \sum_{i=1}^p\gamma_{ji}^2
    =
    (\Sigma^{-1})_{jj}
    =
    \mathrm{VIF}_j.
\]

By the definition of RegPA,
\[
    \mathrm{RegPA}_{ij}
    =
    \frac{\gamma_{ij}^2}
         {\sum_{\ell=1}^p\gamma_{\ell j}^2}
    =
    \frac{\gamma_{ij}^2}{\mathrm{VIF}_j},
\]
which proves Eq.~\eqref{eq:regpa-vif}. Similarly,
\[
    w_{ij}
    =
    \frac{\gamma_{ij}^2}{\mathrm{VIF}_i}
    \in[0,1],
    \qquad
    \sum_{j=1}^p w_{ij}=1.
\]
Hence,
\[
    r_i=
    \sum_{j=1}^p
    \frac{\gamma_{ij}^2}{\mathrm{VIF}_j}
    =
    \sum_{j=1}^p
    w_{ij}
    \frac{\mathrm{VIF}_i}{\mathrm{VIF}_j},
\]
which proves Eq.~\eqref{eq:regpa-row-sum-vif}.

Now let
\[
    k\in\arg\max_{i=1,\ldots,p}\mathrm{VIF}_i.
\]
Then $\mathrm{VIF}_k/\mathrm{VIF}_j\geq1$ for every $j$, so
\[
    r_k
    =
    \sum_{j=1}^p
    w_{kj}
    \frac{\mathrm{VIF}_k}{\mathrm{VIF}_j}
    \geq
    \sum_{j=1}^p w_{kj}
    =
    1.
\]
The inequality is strict if $\gamma_{kj}\neq0$ for at least one $j$ satisfying $\mathrm{VIF}_j<\mathrm{VIF}_k$. By Proposition~\ref{prop:expected-importance}, this implies
\[
    \mathbb{E}_{\beta_0\sim F}
    \left[
        D_{\mathrm{RegPA},k}(\Sigma,\beta_0)
    \right]
    =
    \frac{r_k}{p}
    >
    \frac{1}{p}.
\]

Finally, because every RegPA column sums to one,
\[
    \sum_{i=1}^p r_i
    =
    \sum_{i=1}^p\sum_{j=1}^p\mathrm{RegPA}_{ij}
    =
    p.
\]
Thus, if $r_k>1$, at least one other row-sum must satisfy $r_\ell<1$. Applying Proposition~\ref{prop:expected-importance} again gives
\[
    \mathbb{E}_{\beta_0\sim F}
    \left[
        D_{\mathrm{RegPA},\ell}(\Sigma,\beta_0)
    \right]
    <
    \frac{1}{p},
\]
which completes the proof.
\end{proof}
\subsection{Proof of Proposition~\ref{prop:compound-symmetry}}
\label{app:proof-compound-symmetry}

\begin{proof}
Let
\[
    \lambda_1=1+(p-1)\rho,
    \qquad
    \lambda_2=1-\rho,
    \qquad
    \tau=\tau_\rho.
\]
The eigenspaces spanned by $1_p$ and orthogonal to $1_p$ have eigenvalues $\lambda_1$ and $\lambda_2$, respectively. Therefore,
\begin{equation}
    \Sigma^{1/2}
    =
    \sqrt{\lambda_2}I_p+\tau J_p,
    \qquad
    \Sigma^{-1/2}
    =
    \frac{1}{\sqrt{\lambda_2}}I_p
    -
    \frac{\tau}{\sqrt{\lambda_1\lambda_2}}J_p.
    \label{eq:cs-square-roots}
\end{equation}

Under Johnson's transformation, $Z=X\Sigma^{-1/2}$, so
\[
    X^\top Z
    =
    \Sigma^{1/2}.
\]
The entries of $\Sigma^{1/2}$ are the correlations between the original and orthogonal predictors. Moreover, the squared entries in each column sum to
\[
    \left[(\Sigma^{1/2})^\top\Sigma^{1/2}\right]_{jj}
    =
    \Sigma_{jj}
    =
    1.
\]
Hence, for $i\neq j$,
\[
    \mathrm{CorPA}_{ij}
    =
    \tau^2
    =
    \alpha_C.
\]

The regression-coefficient matrix of $Z$ on $X$ is $\Gamma_Z=\Sigma^{-1/2}$. Under compound symmetry,
\[
    \mathrm{VIF}_j
    =
    (\Sigma^{-1})_{jj}
    =
    \frac{h_\rho}{\lambda_1\lambda_2},
    \qquad j=1,\ldots,p.
\]
The off-diagonal entries of $\Gamma_Z$ equal $-\tau/\sqrt{\lambda_1\lambda_2}$, and therefore
\[
    \mathrm{RegPA}_{ij}
    =
    \frac{
        \tau^2/(\lambda_1\lambda_2)
    }{
        h_\rho/(\lambda_1\lambda_2)
    }
    =
    \frac{\tau^2}{h_\rho}
    =
    \alpha_R,
    \qquad i\neq j.
\]

It remains to derive GDA. Fix $i\neq j$ and treat $z_j$ as the response. By permutation symmetry, $\mathrm{GDA}_{ij}=D_i(z_j,X)$ is identical for all $i\neq j$; denote this common value by $\alpha_G$.

For a subset of $s$ predictors, its correlation matrix has inverse
\[
    \Sigma_s^{-1}
    =
    \frac{1}{1-\rho}
    \left[
        I_s
        -
        \frac{\rho}{1+(s-1)\rho}J_s
    \right],
\]
where $J_s=1_s1_s^\top$. Let $R_0(s)$ denote $R^2_{z_j\cdot X_S}$ when $|S|=s$ and $j\notin S$, and let $R_1(s)$ denote the corresponding value when $j\in S$. For a subset $S$, let
\[
    w_S=(\Sigma^{1/2})_{S,j}
\]
denote the vector of correlations between $z_j$ and the predictors indexed by $S$. Equation~\eqref{eq:cs-square-roots} gives
\begin{align}
    R_0(s)
    &=
    \frac{\tau^2s}{1+(s-1)\rho},
    \nonumber\\
    R_1(s)
    &=
    \frac{1}{\lambda_2}
    \left[
        (\sqrt{\lambda_2}+\tau)^2
        +(s-1)\tau^2
        -
        \frac{
            \rho(\sqrt{\lambda_2}+s\tau)^2
        }{
            1+(s-1)\rho
        }
    \right].
    \label{eq:cs-subset-r2}
\end{align}

Define
\[
    q_s=1+(s-1)\rho,
    \qquad
    t_s=\sqrt{\lambda_2}+s\tau,
\]
and note the identity
\begin{equation}
    \rho
    =
    \tau
    \left(
        \sqrt{\lambda_1}+\sqrt{\lambda_2}
    \right),
    \label{eq:cs-rho-identity}
\end{equation}
which follows from $\lambda_1-\lambda_2=p\rho$ and the definition of $\tau$. Let
\[
    \Delta_0(s)=R_0(s+1)-R_0(s),
    \quad s=0,\ldots,p-2,
    \qquad
    \Delta_1(s)=R_1(s+1)-R_1(s),
    \quad s=1,\ldots,p-1.
\]
Since $q_{s+1}=q_s+\rho$ and $q_s-s\rho=1-\rho=\lambda_2$,
\[
    \Delta_0(s)
    =
    \tau^2
    \left[
        \frac{s+1}{q_{s+1}}-\frac{s}{q_s}
    \right]
    =
    \tau^2\,
    \frac{(s+1)q_s-sq_{s+1}}{q_sq_{s+1}}
    =
    \frac{\tau^2\lambda_2}{q_sq_{s+1}}.
\]
For $\Delta_1(s)$, note that $1_s^\top w_S=t_s$ when $j\in S$, so that Eq.~\eqref{eq:cs-subset-r2} can be written as $R_1(s)=\lambda_2^{-1} \bigl[(\sqrt{\lambda_2}+\tau)^2+(s-1)\tau^2-\rho t_s^2/q_s\bigr]$. Using $t_{s+1}=t_s+\tau$ and $q_{s+1}=q_s+\rho$,
\[
    \Delta_1(s)
    =
    \frac{1}{\lambda_2}
    \left[
        \tau^2
        -\rho
        \left(
            \frac{t_{s+1}^2}{q_{s+1}}
            -\frac{t_s^2}{q_s}
        \right)
    \right]
    =
    \frac{1}{\lambda_2}
    \cdot
    \frac{
        \tau^2q_sq_{s+1}
        -\rho
        \bigl(
            2\tau q_st_s+\tau^2q_s-\rho t_s^2
        \bigr)
    }{q_sq_{s+1}},
\]
and expanding the numerator yields the perfect square
\[
    \tau^2q_s^2-2\tau\rho\,q_st_s+\rho^2t_s^2
    =
    \left(
        \tau q_s-\rho t_s
    \right)^2.
\]
Crucially, $\tau q_s-\rho t_s$ does not depend on $s$: by Eq.~\eqref{eq:cs-rho-identity},
\[
    \tau q_s-\rho t_s
    =
    \tau(1-\rho)-\rho\sqrt{\lambda_2}
    =
    \tau\lambda_2
    -\tau
    \left(
        \sqrt{\lambda_1}+\sqrt{\lambda_2}
    \right)
    \sqrt{\lambda_2}
    =
    -\tau\sqrt{\lambda_1\lambda_2}.
\]
Hence $(\tau q_s-\rho t_s)^2=\tau^2\lambda_1\lambda_2$, and
\begin{equation}
    \Delta_0(s)
    =
    \frac{\tau^2\lambda_2}
         {q_s(q_s+\rho)},
    \qquad
    \Delta_1(s)
    =
    \frac{\tau^2\lambda_1}
         {q_s(q_s+\rho)}.
    \label{eq:cs-marginal-increments}
\end{equation}
In particular, $\Delta_1(s)/\Delta_0(s)=\lambda_1/\lambda_2$ for every $s$: the marginal contribution of $x_i$ to explaining $z_j$ is exactly $\lambda_1/\lambda_2$ times larger when the subset already contains $x_j$ than when it does not, uniformly over subset sizes.

Among the subsets of size $s$ excluding $i$, the proportions that contain and do not contain $j$ are $s/(p-1)$ and $(p-1-s)/(p-1)$, respectively. Consequently,
\begin{align}
    \alpha_G
    &=
    \frac{1}{p(p-1)}
    \left[
        \sum_{s=0}^{p-2}
        (p-1-s)\Delta_0(s)
        +
        \sum_{s=1}^{p-1}
        s\Delta_1(s)
    \right]
    \nonumber\\
    &=
    \frac{\tau^2}{p(p-1)}
    \sum_{s=0}^{p-1}
    \frac{
        s\lambda_1+(p-1-s)\lambda_2
    }{
        q_s(q_s+\rho)
    },
\end{align}
where the second equality adds the two vanishing boundary terms. The summand satisfies
\[
    \frac{
        s\lambda_1
        +(p-1-s)\lambda_2
    }{
        q_s(q_s+\rho)
    }
    =
    \frac{p-1}{1+s\rho}
    +
    s
    \left(
        \frac{1}{q_s}
        -
        \frac{1}{q_s+\rho}
    \right).
\]
The second term is evaluated by shifting the summation index:
\[
    \sum_{s=0}^{p-1}
    s\left(\frac{1}{q_s}-\frac{1}{q_{s+1}}\right)
    =
    \sum_{s=1}^{p-1}\frac{s}{q_s}
    -\sum_{r=1}^{p}\frac{r-1}{q_r}
    =
    \sum_{s=1}^{p-1}\frac{1}{q_s}
    -\frac{p-1}{q_p}
    =
    \sum_{s=0}^{p-2}\frac{1}{1+s\rho}
    -\frac{p-1}{\lambda_1},
\]
since $\{q_s\}_{s=1}^{p-1}=\{1,\,1+\rho,\ldots,\,1+(p-2)\rho\}$ and $q_p=\lambda_1$. The first term contributes
\[
    \sum_{s=0}^{p-1}\frac{p-1}{1+s\rho}
    =
    (p-1)\sum_{s=0}^{p-2}\frac{1}{1+s\rho}
    +\frac{p-1}{\lambda_1},
\]
and adding the two parts cancels the $\lambda_1$ terms, giving
\[
    \sum_{s=0}^{p-1}
    \frac{
        s\lambda_1
        +(p-1-s)\lambda_2
    }{
        q_s(q_s+\rho)
    }
    =
    p\sum_{s=0}^{p-2}\frac{1}{1+s\rho}.
\]
Therefore,
\[
    \alpha_G
    =
    \frac{\tau^2}{p-1}
    \sum_{s=0}^{p-2}\frac{1}{1+s\rho}
    =
    \kappa_\rho\tau^2.
\]

All three matrices have unit column-sums, so their diagonal entries are $1-(p-1)\alpha_M$, establishing Eq.~\eqref{eq:cs-reallocation-matrices}, and Eq.~\eqref{eq:cs-off-diagonal-relations} follows by taking ratios of the off-diagonal entries. At $\rho=0$, $\tau_\rho=0$, so all three matrices equal $I_p$.

\end{proof}
\subsection{Proof of Corollary~\ref{cor:cs-no-bias}}
\label{app:proof-cs-no-bias}

\begin{proof}
By Proposition~\ref{prop:compound-symmetry}, each of the three reallocation matrices has the form
\[
    A(\alpha)
    =
    (1-p\alpha)I_p+\alpha J_p.
\]
Hence,
\[
    A(\alpha)1_p
    =
    (1-p\alpha)1_p+\alpha p1_p
    =
    1_p.
\]
Because $A(\alpha)$ is symmetric, it also satisfies $1_p^\top A(\alpha)=1_p^\top$. Since each matrix has unit row- and column-sums, and is non-negative by construction, it is doubly stochastic. Proposition~\ref{prop:expected-importance} then gives
\[
    \mathbb{E}_{\beta_0\sim F}
    \left[
        D_{A,i}(\Sigma,\beta_0)
    \right]
    =
    \frac{1}{p}.
\]

The inverse of the compound-symmetric correlation matrix is
\[
    \Sigma^{-1}
    =
    \frac{1}{1-\rho}
    \left[
        I_p
        -
        \frac{\rho}{1+(p-1)\rho}J_p
    \right].
\]
Therefore,
\begin{align*}
    \mathrm{VIF}_i
    &=
    (\Sigma^{-1})_{ii}\\
    &=
    \frac{1}{1-\rho}
    \left[
        1-\frac{\rho}{1+(p-1)\rho}
    \right]\\
    &=
    \frac{
        1+(p-2)\rho
    }{
        [1+(p-1)\rho](1-\rho)
    },
\end{align*}
which is identical for every predictor.
\end{proof}
\subsection{Proof of Corollary~\ref{cor:cs-shrinkage}}
\label{app:proof-cs-shrinkage}

\begin{proof}
For a reallocation matrix
\[
    A(\alpha)
    =
    (1-p\alpha)I_p+\alpha J_p,
\]
the normalized ORM importance vector is $D_{A(\alpha)}=A(\alpha)\eta$. Since $1_p^\top\eta=1$, we have $J_p\eta=1_p$, and therefore
\[
    D_{A(\alpha)}
    =
    (1-p\alpha)\eta+\alpha1_p.
\]
Because $(1-p\alpha)\mu_p+\alpha1_p=\mu_p$, it follows that
\[
    D_{A(\alpha)}-\mu_p
    =
    (1-p\alpha)(\eta-\mu_p),
\]
which proves Eq.~\eqref{eq:cs-shrinkage-representation}.
 
We next order the off-diagonal values. By Proposition~\ref{prop:compound-symmetry}, $\alpha_G/\alpha_C=\kappa_\rho$ and $\alpha_G/\alpha_R=h_\rho\kappa_\rho$. For $0<\rho<1$,
\[
    \kappa_\rho
    =
    \frac{1}{p-1}
    \sum_{s=0}^{p-2}
    \frac{1}{1+s\rho}
    <1,
\]
because every term in the average except the first is strictly less than one, while
\[
    h_\rho\kappa_\rho
    =
    \frac{1}{p-1}
    \sum_{s=0}^{p-2}
    \frac{h_\rho}{1+s\rho}
    >1,
\]
because each ratio $h_\rho/(1+s\rho)$ is at least one, with strict inequality for at least one $s$. Hence $\alpha_R<\alpha_G<\alpha_C$. For $-\tfrac{1}{p-1}<\rho<0$, both arguments reverse, giving $\alpha_C<\alpha_G<\alpha_R$.
 
Furthermore, $p\alpha_M<1$ for every $M\in\{C,G,R\}$ and every $\rho\neq0$. Writing $\lambda_1=1+(p-1)\rho$ and $\lambda_2=1-\rho$, we have $\lambda_1\lambda_2>0$ and, since $(p-2)\lambda_1>0$ and $h_\rho>0$ on the parameter range,
\[
    \left(
        \sqrt{\lambda_1}-\sqrt{\lambda_2}
    \right)^2
    <
    \lambda_1+\lambda_2
    =
    2+(p-2)\rho
    <
    p,
    \qquad
    \lambda_1+\lambda_2
    <
    \lambda_1+\lambda_2+(p-2)\lambda_1
    =
    p\,h_\rho.
\]
The first chain gives $p\alpha_C=(\sqrt{\lambda_1}-\sqrt{\lambda_2})^2/p<1$, and the second gives $p\alpha_R=(\sqrt{\lambda_1}-\sqrt{\lambda_2})^2/(p\,h_\rho)<1$; since $\alpha_G$ always lies between $\alpha_R$ and $\alpha_C$, we conclude $0<1-p\alpha_M<1$ for all three methods.
 
Now suppose $0<\rho<1$. Then
\[
    0
    <
    1-p\alpha_C
    <
    1-p\alpha_G
    <
    1-p\alpha_R,
\]
and if $\eta\neq\mu_p$, absolute homogeneity of any norm applied to Eq.~\eqref{eq:cs-shrinkage-representation} yields
\[
    \lVert D_C-\mu_p\rVert
    <
    \lVert D_G-\mu_p\rVert
    <
    \lVert D_R-\mu_p\rVert.
\]
For $\rho<0$, the ordering of the $\alpha_M$ reverses, and hence so does the norm ordering.
 
Finally,
\begin{align*}
    D_C-D_G
    &=
    \left[
        (1-p\alpha_C)-(1-p\alpha_G)
    \right]\eta
    +
    (\alpha_C-\alpha_G)1_p\\
    &=
    p(\alpha_C-\alpha_G)
    \left(
        \mu_p-\eta
    \right),
\end{align*}
which proves Eq.~\eqref{eq:cs-correlation-excess-shrinkage}.
\end{proof}

\bibliographystyle{unsrtnat}
\bibliography{references} 

\end{document}